\documentclass{amsart}
\usepackage{amsbsy,amssymb,amscd,amsfonts,latexsym,amstext,delarray,
amsmath,amsthm,graphicx}

\input xypic
\usepackage[all]{xy}

\newtheorem{thm}{Theorem}[section]
\newtheorem{prop}[thm]{Proposition}
\newtheorem{cor}[thm]{Corollary}
\newtheorem{lem}[thm]{Lemma}

\newtheorem{defn}[thm]{Definition}
\newtheorem{rem}[thm]{Remark}
\newtheorem{ex}[thm]{Example}

\numberwithin{equation}{section}

\newcommand{\Abb}{{\mathbb{A}}}

\newcommand{\Lbb}{{\mathbb{L}}}

\newcommand{\Pbb}{{\mathbb{P}}}

\newcommand{\Tbb}{{\mathbb{T}}}

\newcommand{\cC}{{\mathcal{C}}}

\newcommand{\cN}{{\mathcal{N}}}

\newcommand{\cT}{{\mathcal{T}}}

\newcommand{\cV}{{\mathcal{V}}}
\newcommand{\cW}{{\mathcal{W}}}
\newcommand{\cX}{{\mathcal{X}}}
\newcommand{\cY}{{\mathcal{Y}}}
\newcommand{\cZ}{{\mathcal{Z}}}

\def\bA{{\mathbb A}}

\def\bG{{\mathbb G}}

\def\bL{{\mathbb L}}

\def\bT{{\mathbb T}}

\def\A{{\mathbb A}}
\def\C{{\mathbb C}}

\def\N{{\mathbb N}}
\renewcommand{\P}{{\mathbb P}}

\def\Z{{\mathbb Z}}
\def\R{{\mathbb R}}
\def\K{{\mathbb K}}

\def\fA{{\mathfrak A}}

\newcommand{\bc}[2]{{}^{#1}{#2}}

\title{A motivic approach to phase transitions in Potts models}
\author{Paolo Aluffi}
\address{Department of Mathematics \\ Florida State University \\ Tallahassee \\
FL 32306 \\ USA}
\email{aluffi@math.fsu.edu}
\author{Matilde Marcolli}
\address{Division of Physics, Mathematics and Astronomy \\ California Institute of Technology \\  Pasadena \\ CA 91125 \\ USA}
\email{matilde@caltech.edu}

\begin{document}
\maketitle

\begin{abstract}
We describe an approach to the study of phase transitions in Potts models
based on an estimate of the complexity of the locus of real zeros of the
partition function, computed in terms of  the classes in the Grothendieck ring  
of the affine algebraic varieties defined by the vanishing of the
multivariate Tutte polynomial. We give completely explicit calculations for
the examples of the chains of linked polygons and of the graphs obtained
by replacing the polygons with their dual graphs. These are based on a
deletion--contraction formula for the Grothendieck classes and on generating
functions for splitting and doubling edges.
\end{abstract}

\tableofcontents

\section{Introduction}

It is well known that the partition function of a Potts model with $q$ spin states
on a graph $G$ is given by the value at $q$ of the multivariate Tutte polynomial
of the graph, a famous combinatorial invariant of graphs. 
The problem of phase transitions in Potts models is related to
the behavior of the sets of complex zeros and real zeros of these polynomials,
see for instance \cite{Sokal}.

In this paper we propose a new approach, based on algebraic geometry,
and especially on motivic invariants such as classes in the Grothendieck
ring of varieties, to study how the set of zeros of the partition function of
a Potts model changes for a nested family of finite graphs that grow in
size to approximate an infinite graph. 

We aim at estimating the ``topological complexity" of the set of real
zeros by computing its Euler characteristic with compact support,
which is also known to provide an estimate of the algorithmic complexity
of the real algebraic set. In order to compute this invariant and estimate its
growth over certain explicit families of graphs, we use the fact that 
the invariant is not just topological but also ``motivic", which means that
it defines a ring homomorphism from the Grothendieck ring of varieties
to the integers.

Thus, we proceed first to compute explicitly the classes in the Grothendieck
ring of the varieties defined by the zeros of the Potts models partition
functions, following techniques recently developed to treat a similar
problem arising in perturbative quantum field theory, for algebraic
varieties associated to the parametric form of Feynman integrals.
The varieties arising in quantum field theory can be viewed, up to
a duality, as a limit case of the ones arising from the Potts model
partition functions.

We first prove an algebro-geometric inclusion-exclusion
formula that relates the classes for a given graph to those of the graphs
obtained by deleting or contracting one edge in the graph, and a more
complicated algebro--geometric term, which is the variety defined by
the intersection of the varieties of the deletion and the contraction. This
formula is similar to an analogous result proved in \cite{AluMa3} for
the varieties arising from Feynman diagrams in quantum field theory.

We then show that, when iterating simple operations on graphs, such as
splitting an edge or doubling an edge, the more complicated
term in the algebro--geometric deletion contraction formula can
be simplified due to cancellations in the Grothendieck ring
and the resulting operation can be described completely in
terms of varieties associated to purely combinatorial data on the graph.
The corresponding recursive relation leads to explicit and remarkably simple generating 
functions for the classes of graphs obtained from a given
graph by multiple edge splittings or edge doublings. 

We use the formulae obtained in this way to write explicitly the
classes in the {Grothen\-dieck} ring for the loci of zeros of the
partition function on Potts models over graphs given by chains
of linked polygons. This class of Potts models was already 
studied by different techniques in the literature (see for
instance \cite{ShTs} and references therein). Similarly, we
compute the Grothendieck classes explicitly for similar 
chains where the polygons are replaced by  their dual graphs,
the banana graphs. 

For these illustrative examples, we then use the expression for
the Grothendieck class to compute explicitly the Euler
characteristic with compact support, and we show that it
grows exponentially as the graphs grow in size, thus
estimating the corresponding growth in complexity of the
set of real zeros of the partition function.

\section{Potts models, multivariate Tutte polynomial and hypersurfaces}

We recall some basic facts and terminology about Potts model that we need in the rest
of the paper. For a more detailed introduction to partition functions of Potts models and
their relation to graph combinatorics, we refer the reader to \cite{Sokal}, see also
the survey \cite{BEMPS}.

\subsection{Multivariate Tutte polynomial}

The multivariate Tutte polynomial of a finite graph $G$ is defined as follows. Let $V=V(G)$
be the set of vertices of $G$ and $E=E(G)$ the set of edges. We do not assume $G$ to
be connected. One assigns to each edge $e\in E$ a variable $t_e$ and then considers
the polynomial
\begin{equation}\label{multiTutte}
Z_G(q,t) = \sum_{G'\subseteq G} q^{k(G')} \prod_{e\in E(G')} t_e,
\end{equation}
where $k(G')=b_0(G')$ is the number of connected components, and 
the sum is over all subgraphs $G' \subseteq G$ that have the same
number of vertices $V(G')=V(G)$ of $G$. Each subgraph $G'$ therefore
corresponds to a choice of a subset $A\subseteq E(G)$ of edges of $G$,
so that $E(G')=A$. The variables $t_e$ and the additional variable $q$ 
are commuting variables.

\medskip

A detailed account of the relation between the multivariate Tutte polynomial and the 
physics of Potts models is given in the survey \cite{Sokal}.  To briefly recall the main
point, in the case where $q$ is a positive integer, one considers a $q$-state model
on a graph $G$, where each vertex carries a ``spin" that can take $q$ possible values
(the case $q=2$ recovers the usual $\pm 1$ states of spin). We let $\fA$ be the set
of cardinality $q$ of the possible spin states. A state of the system is
an assignment of a spin state to each vertex and the energy $H$ of a state is the sum
over all edges of the graph of a quantity that is equal to zero if the spins assigned to
the two endpoints of the edge are different and equal to an assigned value $-J_e$
if they are the same.  With the notation $t_e = e^{\beta J_e}-1$, where $\beta$ is
the thermodynamic parameter (an inverse temperature up to the Boltzmann constant),
one has $t_e\geq 0$ in the ferromagnetic case ($J_e\geq 0$) and $-1\leq t_e\leq 0$
in the antiferromagnetic case $-\infty \leq J_e\leq 0$. The partition function of
the system is then the sum over all the possible states of the corresponding Boltzmann
weight $e^{-\beta H}$, with $H$ the energy of that state. This gives
\begin{equation}\label{PottsZ}
Z_G(q,t)= \sum_{\sigma: V(G)\to \fA} \,\,\,\, \prod_{e\in E(G)} (1+ t_e \delta_{\sigma(v),\sigma(w)}),
\end{equation}
where the sum is over all maps of vertices to spin states, $v$ and $w$ are the endpoints
$\partial(e)=\{ v,w \}$, and $\delta$ is the Kronecker delta.

It was shown by Fortuin--Kasteleyn \cite{FoKa} that \eqref{PottsZ} is the
restriction to positive integer values of $q$ of a polynomial function in $q$ and that
this polynomial is, in fact, the multivariate Tutte polynomial \eqref{multiTutte}.

\subsection{Deletion--contraction relation}

As in the case of the ordinary Tutte polynomial, the multivariate Tutte polynomial
\eqref{multiTutte} satisfies a deletion--contraction relation. Namely, given an edge
$e\in E(G)$, let $G\smallsetminus e$ be the graph obtained by deleting the edge
$e$ and let $G/e$ be the graph obtained by contracting it. One has the following
formula.

\begin{rem}\label{TutteDelConRem} {\rm
The polynomial
\eqref{multiTutte} satisfies
\begin{equation}\label{multidelcon}
Z_G(q,t)= Z_{G\smallsetminus e}(q,\hat t) + t_e Z_{G/e}(q,\hat t),
\end{equation}
where $\hat t$ consists of the edge variables with $t_e$ removed. 
The deletion--contraction relation \eqref{multidelcon} covers all cases, including
those where the edge $e$ is a bridge or a looping edge.}
\end{rem}

\subsection{The problem of phase transitions}

Zeros of the multivariate Tutte polynomial are of special interest in relation
to the problem of phase transitions of the statistical mechanical system
described by the Potts model. In fact, the partition function $Z_G(q,t)$ becomes
the normalization factor of the probability distribution on the set of all possible
states of the system, and zeros of $Z_G(q,t)$ would signal the presence of a phase 
transition in the system, for a specific choice of parameters $J_e$ and $q$, and for
certain values of the inverse temperature $\beta$. 

In the ferromagnetic case, the physical case of interest is where the variables
$t_e \in \R_+$. For $q\geq 1$, the polynomial $Z_G(q,t)$ does not have zeros
in that domain. The antiferromagnetic case with $-1\leq t_e\leq 0$ is more 
interesting, and various results on zero-free regions are given in \cite{JackSokal}.

Even when there are no zeros in the region of direct physical interest, it
is well known (see \cite{JackSokal}, \cite{Sokal}) that studying the {\em complex}
zeros of the polynomials $Z_G(q,t)$ can provide useful information on phase
transitions, not for a single graph $G$ itself, but for a family of finite graphs $G_n$,
such that $G_\infty=\cup_n G_n$ determines an infinite graph on which one still considers
a statistical mechanical system obtained as a thermodynamic limit of the finite systems.
If loci of complex zeros of the $Z_G(q,t)$ can approach the real locus in the limit,
this will result in the presence of phase transitions for the system on $G_\infty$.

Thus, the geometric problem we concentrate on is to understand and
estimate how the loci of complex and real zeros, respectively, of the
Potts model partition function change over certain families of finite
graphs $\{ G_n \}$ as above.

Our point of view, in this paper, is to approach this problem from an
algebro--geometric and motivic point of view, inspired by recent developments
on motivic properties of the loci of zeros of the closely related Kirchhoff graph
polynomials in the setting of perturbative quantum field theory, see \cite{AluMa3},
\cite{AluMa2}, \cite{AluMa}, \cite{BEK}, \cite{BK1}, \cite{Mar-book}.

\subsection{Potts model hypersurfaces}

Studying the zeros of the polynomial $Z_G(q,t)$, means
understanding the geometry of the hypersurface defined by the 
equation $Z_G(q,t)=0$. Since the polynomial is not homogeneous in its
variables, it does not define a projective hypersurface (unlike the case
of the graph polynomials in quantum field theory), but it does define
an affine hypersurface, which we refer to as {\em the Potts model
hypersurface}. In fact, we consider two types of hypersurfaces
associated to the Potts model, one where the parameter $q$ is
treated as a variable $q\in \A$ along with the other edge variables $t_e$,
and one where one specializes to a fixed value of $q$.

\begin{defn}\label{PottsModHypDef}
Suppose given a finite graph $G$, with set of vertices $V(G)$ and set of edges $E(G)$.
Let $\cZ_G$ be the hypersurface in affine space $\A^{\# E(G)+1}$ defined by 
\begin{equation}\label{PottsModHyp}
\cZ_G := \{ (q,t)\in \A^{\# E(G)+1} \, |\, Z_G(q,t)=0 \}.
\end{equation}
For a fixed value of $q\in \A$, the hypersurface $\cZ_{G,q}$ in $\A^{\# E(G)}$ is given by
\begin{equation}\label{qPottsModHyp}
\cZ_{G,q} := \{ t\in \A^{\# E(G)} \, |\, Z_G(q,t)=0 \}.
\end{equation}
\end{defn}

The hypersurface $\cZ_{G,q}$ is therefore a slice of $\cZ_G$ with the 
hyperplane in $\A^{\# E(G)+1}$ of fixed $q$-coordinate. In the physically
significant cases, one wants to study the complex and the real zeros of the hypersurface
$\cZ_{G,q}$ where $q\in \N$ is a positive integer corresponding to the
number of spin-states of the Potts model.

\begin{defn}\label{virphasetrans}
For a finite graph $G$, the {\em virtual phase transitions} of the Potts model
are the real points  $\cZ_{G}(\R)$ of the algebraic variety $\cZ_G$. For a fixed
$q$, the virtual phase transitions are the points of the real locus $\cZ_{G,q}(\R)$ 
of the variety $\cZ_{G,q}$.
\end{defn}

We distinguish here between virtual phase transitions (all real zeros
of the polynomial $Z_G(q,t)$) and the actual physical phase transitions,
which would be constrained by the additional requirement that $q\in\N$
and the edge variables are $t_e\geq 0$ in the ferromagnetic case, or
$-1\leq t_e \leq 0$ in the antiferromagnetic case. Thus,  for example,
in the case of a finite graph, even if there are no physical phase transitions
in the ferromagnetic case, there can still be a non-empty set of
virtual phase transitions.

\subsection{The analogy with Quantum Field Theory}

In perturbative quantum field theory, the parametric form of
Feynman integrals for massless scalar field theories can be
expressed as a (possibly divergent) integral of an algebraic
differential form on the complement of an algebraic hypersurface
defined by the vanishing of a polynomial associated to the
graph, the first Kirchhoff polynomial given by
\begin{equation}\label{Kirchhoff}
\Psi_G(t)= \sum_{T\subseteq G} \prod_{e\notin T} t_e,
\end{equation}
where $t=(t_1,\ldots,t_n)$ are variables assigned to the
edges of the graph and $T$ runs over maximal spanning
forests, that is, subgraphs of $G$ with $V(T)=V(G)$, which 
are forests with $b_0(T)=b_0(G)$. (Note, the terminology
``spanning forest" is used elsewhere for what we refer to here
as ``maximal spanning forests".)

In the literature on motivic aspects of Feynman integrals \cite{BeBro},
\cite{Stanley}, \cite{Stem}, it is also common to consider the related polynomial
\begin{equation}\label{dualKirchhoff}
\Phi_G(t) = \sum_{T\subseteq G} \prod_{e\in T} t_e, 
\end{equation}
where the sum is, as above, over the maximal spanning
forests, but the product is on edges in the forest, instead
of edges in the complement. 

\begin{defn}\label{GraphHypDef}
We denote by $\cX_{G}\subset \A^{\# E(G)}$ the affine hypersurface defined by
the polynomial \eqref{Kirchhoff} and by $\bar \cX_{G} \subset \P^{\# E(G)-1}$
the corresponding projective hypersurface.

Similarly, we denote by $\cY_{G}\subset \A^{\# E(G)}$ the affine hypersurface defined by
the polynomial \eqref{dualKirchhoff} and by $\bar \cY_{G} \subset \P^{\# E(G)-1}$
the corresponding projective hypersurface.
\end{defn}

\begin{rem}\label{PsiPhiRem}{\rm 
One obtains $\Psi_G(t)$ from $\Phi_G(t)$ by dividing by $\prod_{e\in E(G)} t_e$
and changing variables by the transformation $t_e \mapsto 1/t_e$. }
\end{rem}

For a planar graph
this operation relates the Kirchhoff polynomial of a graph with that of a dual
graph, see the discussion on the Cremona transformation in \cite{AluMa}.

\begin{rem}\label{PhiTutteRem} {\rm
It is also well known (see for instance \cite{KRTW}, \cite{Sokal}) that the
graph polynomial $\Phi_G(t)$ of \eqref{dualKirchhoff} can be recovered
from the multivariate Tutte polynomial by the following operations:
\begin{enumerate}
\item Clear an overall factor of $q^{k(G)}$ with $k(G)=b_0(G)$ the
number of connected components, that is, consider the normalized
Potts partition function 
\begin{equation}\label{tildeZq}
\tilde Z_G(q,t) = q^{-k(G)} Z_G(q,t).
\end{equation}
\item Take the evaluation $\tilde Z_G(q,t)|_{q=0}$. This corresponds to
a sum on subgraphs $G'$ with $k(G')=k(G)$.
\item Of this take then the homogeneous piece with the lowest degree in the
$t=(t_e)$ variables. This corresponds to the sum over maximal spanning
forests, that is, to the polynomial $\Phi_G(t)$. 
\end{enumerate} }
\end{rem}

In the context of quantum field theory, Tutte polynomials
can also occur, for example where one considers scalar field theories
on noncommutative spacetimes as in \cite{KRTW}.

\subsection{The tangent cone}

Consider an affine hypersurface $X \subset \A^N$ 
given by the vanishing $X=\{ t \in \A^N \,|\, P(t) =0 \}$ of a 
(possibly non-homogeneous) polynomial $P(t)$ whose leading term $P_k(t)$
(the term of lowest order in the variables $t=(t_i)$)  is
of some degree $k\geq 1$.  Then the {\em tangent cone\/} to $X$ at the origin,
$\cT\cC(X)=\cT\cC_0(X)$ is also a hypersurface in $\A^N$, given by
\begin{equation}\label{tangentconeX}
\cT\cC(X) =\{ t\in \A^N \,|\, P_k(t) =0 \}, 
\end{equation}
the zero locus of the homogeneous polynomial $P_k$. This corresponds
to the {\em normal cone\/} $\cN\cC_0(X)$ for the subscheme given by the origin
$0 \subset X$. The construction known as {\em deformation to the normal cone} 
provides a very useful algebro-geometric replacement of the notion of tubular
neighborhoods of embedded subvarieties (subschemes), see \cite{Ful}.

In the case of a closed subscheme $Y\subset X$, one blows up 
the locus $Y \times \{ 0 \}$ inside $X \times \P^1$. One then considers
the complement 
$$ \tilde X_Y :=Bl_{Y\times \{ 0 \}}(X \times \P^1) \smallsetminus Bl_Y(X). $$
One obtains in this way a fibration $\tilde X_Y \to \P^1$, 
whose general fiber (away from $0$) is naturally isomorphic to $X$,
while the special fiber over zero is the
tangent cone $\cT\cC_Y(X)$.
This has indeed the effect of deforming $X$ to the tangent cone $\cT\cC_Y(X)$.

For a finite graph $G$, we denote by $P_G$ the homogeneous polynomial
\begin{equation}\label{Pleading}
P_G(q,t) = \text{ leading term of } Z_G(q,t),
\end{equation}
in the variables $(q,t)\in \A^{\#E(G)+1}$ and by $\cV_G$ the affine variety
\begin{equation}\label{coneVGvar}
\cV_G =\{ (q,t)\in \A^{\#E(G)+1} \,|\, P_G(q,t)=0 \}.
\end{equation}
Since the polynomial $P_G$ is homogeneous, 
we can also consider the projective hypersurface $\bar \cV_G \subset \P^{\#E(G)}$.
We also consider the affine hypersurfaces 
\begin{equation}\label{YGq}
\cV_{G,q} =\{ t\in \A^{\#E(G)} \,|\, P_G(q,t)=0 \},
\end{equation}
for fixed $q$. These are not homogeneous, in general, except in the case $q=0$.

We then have the following rephrasing of Remark \ref{PhiTutteRem}.

\begin{lem}\label{normconegraph}
The variety $\cV_G$ is the tangent cone of the variety $\cZ_G$ at zero.
It has a component given by the hyperplane $H=\{ q=0 \}$ with
multiplicity equal to $k(G)=b_0(G)$ and another component $\cW_G$,
which intersects the hyperplane $H$ along the graph hypersurface $\cY_G$.
\end{lem}

\proof The first statement follows directly from the definition \eqref{tangentconeX}
of the tangent cone. 

The polynomial $P_G(q,t)$ as in \eqref{Pleading} is the sum of terms
of lowest degree in $Z_G(q,t)$. To see what they parameterize,
note that if the subgraph $G'$ determined by a set of edges $A\subseteq E$ is {\em not\/}
a forest, then one or more of the edges may be removed from $A$ without 
affecting the number of connected components, i.e., the exponent of $q$;
while the degree of the product $\prod t_e$ decreases accordingly.

Further, assume that $A$ is a forest. Then 
\[ k(A)+|A| = \# V(G) . \]
Indeed, this is clear if $A=\emptyset$; and the left-hand-side does not
change if we add an edge connecting vertices without closing cycles
($|A|$ increases by $1$, $k(A)$ decreases by $1$ for each such 
operation). Therefore, all contributions of forests to $Z_G(q,t)$
have degree equal to the number of vertices of $G$, and this is the
lowest possible degree. 

Thus, the polynomial $P_G(q,t)$ collects the contribution of those terms of $Z_G(q,t)$
corresponding to subgraphs that are forests with $V(G')=V(G)$ (spanning forests),
\begin{equation}\label{PGforests}
 P_G(q,t) =
\sum_{G' \subseteq G,\, b_1(G')=0,\, \#V(G')=N} \,\,\, q^{k(G')} \prod_{e\in E(G')} t_e .
\end{equation}
The hypersurface $\cW_G$ is the locus of zeros of the polynomial $Q_G(q,t)$
satisfying $P_G(q,t)= q^{k(G)} Q_G(q,t)$, where $q$ does not divide $Q_G(q,t)$.

The intersection $H\cap \cW_G$ is then given by the locus of zeros of the
polynomial
\begin{equation}\label{QG0t}
 Q_G(0,t) = \sum_{G' \subseteq G, \text{ max forest}} \prod_{e\in E(G')} t_e, 
\end{equation} 
which is the polynomial $\Phi_G(t)$ of \eqref{dualKirchhoff}.
\endproof

\subsection{The Grothendieck ring of varieties and mixed Hodge structures}

The algebraic varieties $\cX_G$ and $\cY_G$ associated to Feynman graphs
have been studied extensively in recent years in terms of their classes in the
Grothendieck ring of varieties, see for instance \cite{BeBro}, \cite{BEK}, \cite{Stanley},
\cite{Stem}. We will be applying here analogous techniques to the Potts model
hypersurfaces $\cZ_G$ and $\cZ_{G,q}$ and apply the results to the problem
of phase transitions. Thus, we recall here a few things about the Grothedieck
ring of varieties, and a result of \cite{AluMa3}, whose analog for Potts models
hypersurfaces we prove in this paper and will be the basis of our motivic
approach to phase transitions.

The Grothendieck ring $K_0(\cV_\K)$ of varieties over a field $\K$ is
generated by isomorphism classes $[X]$ of smooth (quasi)projective varieties 
with the inclusion-exclusion relation
\begin{equation}\label{inclexclrel}
[X] = [Y] + [X\smallsetminus Y]
\end{equation}
for any closed embedding  of a subvariety $Y\subset X$, and with the 
product structure given by $[ X \times Y ] =[X] [Y]$.

In the following, we will be interested in considering the Potts model hypersurfaces 
as varieties defined over $\C$, but we will also be focusing on their real zeros, hence
thinking of them as varieties over $\R$. Thus, in the following we simply write 
$K_0(\cV)$ for the Grothendieck ring, whenever the arguments do not depend on
what field we work over, and we will explicitly mention $\C$ or $\R$ when needed.

The class $[X]$ in the Grothendieck ring is a universal Euler characteristic
for algebraic varieties (see \cite{Bittner}), in the sense that any invariant
of isomorphism classes of algebraic varieties that satisfies the inclusion-exclusion
relation and is multiplicative on products factors through the Grothendieck ring.
These invariants are sometimes called {\em motivic}.

In particular, in the case of complex algebraic varieties and of classes in $K_0(\cV_\C)$,
among these motivic invariants that factor through the Grothendieck
ring we have the usual topological Euler characteristic, but also the virtual Hodge 
polynomials. These will be useful to us in the following so we recall briefly the definition.

The virtual Hodge polynomial of an algebraic variety is defined as 
\begin{equation}\label{HodgePoly}
e(X) (x,y)= \sum_{p,q=0}^d e^{p,q}(X) x^p y^q, \ \ \ \text{ with } \ \ \ 
e^{p,q}(X) =\sum_{k=0}^{2d} (-1)^k h^{p,q}(H^k_c(X)),    
\end{equation}
where, for each pair of integers $(p,q)$, the term $h^{p,q}(H^k_c(X))$ is the Hodge number 
of the mixed Hodge structure on the cohomology with compact support of $X$. If $X$ is
smooth projective, then the virtual Hodge polynomial reduces to the Poincar\'e polynomial, 
with $e^{p,q}(X)=(-1)^{p+q}h^{p,q}(X)$ being the classical pure Hodge numbers. 

The fact that the virtual Hodge polynomial factors through the Grothendieck ring
$K_0(\cV_\C)$ of varieties means that an explicit formula for the class of a variety in the
Grothendieck ring can be used to compute the virtual Hodge polynomial and obtain
some explicit information on the Hodge numbers and the mixed Hodge structure.

\subsection{Virtual Betti numbers of real algebraic varieties}\label{chicSec}

As we mentioned above, in the case of complex algebraic varieties,
the ordinary topological Euler characteristic is  a motivic invariant. This is
not true for real algebraic varieties, as the additive property over closed
embeddings need not be satisfied. However, it is known that there is a
unique motivic invariant that agrees with the topological Euler characteristic
on compact smooth real algebraic varieties and is homeomorphism invariant
(but not homotopy invariant), see \cite{Quarez} and also \cite{Dut}, \cite{Tota}.
It is defined, for any real (semi)algebraic set $S$, as 
\begin{equation}\label{chiBM}
\chi_c(S) = \sum_k (-1)^k b_k^{BM}(S),
\end{equation}
where $b_k^{BM}(S)$ are the Borel--Moore Betti numbers, namely the ranks of
the relative homologies $H_k(\bar S,\infty)$, where $\bar S$ is the Alexandrov
compactification. Equivalently, they are the ranks of the cohomology with compact
support $H^*_c(S)$. 

\begin{ex}\label{chicLT}{\rm 
Let $\bL=[\A^1]$ be the Lefschetz motive, the class of the affine line in $K_0(\cV)$
and let $\bT=[\bG_m]=\bL-1$ be the class of the multiplicative group $\bG_m=\A^1 \smallsetminus \{ 0 \}$.  The topological Euler characteristic $\chi: K_0(\cV_\C)\to \Z$ satisfies $\chi(\bL)=1$
and $\chi(\bT)=0$, while the Euler characteristic with compact support $\chi_c : K_0(\cV_\R)
\to\Z$ satisfies $\chi_c(\bL)=-1$ and $\chi_c(\bT)=-2$.}
\end{ex}

Moreover, it is shown in \cite{McCroPar} that the Betti numbers with $\Z/2\Z$ coefficients 
$b_k(X)=\dim H_k(X,\Z/2\Z)$, defined in the usual way for  
compact smooth real algebraic varieties, extend in a unique way to $K_0(\cV_\R)$,
so that one obtains a ring homomorphism
\begin{equation}\label{virtualPoincare}
\beta: K_0(\cV_\R) \to \Z[t], 
\end{equation}
such that $\beta(X,t)=\sum_k b_k(X) t^k$ for compact smooth varieties. The coefficients
$\beta_k$ of the ring homomorphism $\beta$ are called the {\em virtual Betti numbers}.
They are {\em not} topological invariants. However, they compute the Euler characteristic
\eqref{chiBM}, namely, 
\begin{equation}\label{chicbeta}
\beta(X,-1) = \chi_c (X)
\end{equation}
for all real algebraic varieties $X$, with both equal to the ordinary Euler characteristic $\chi(X)$
in the compact smooth case. Notice that, while $\chi_c(X)$ is the alternating sum of
the ranks of the Borel--Moore homologies, the virtual Betti numbers $\beta_k(X)$
are in general not equal to the Borel--Moore Betti numbers $b_k^{BM}(X)$ (for instance,
the $\beta_k(X)$ can be negative), although their alternating sums agree.

In the case of a compact smooth real algebraic variety, which is the real locus $X(\R)$
of a smooth projective complex algebraic variety $X(\C)$, there are ways to bound the
``topological complexity" of $X(\R)$ in terms of invariants of $X(\C)$, in the form of
Petrovski\u\i--Ole\u\i nik inequalities: for example, for $X(\R)$ a smooth projective real 
algebraic variety of even dimension $n=2p$, one has \cite{Arnold}, \cite{Kharl}
\begin{equation}\label{POineq}
| \chi(X(\R)) -1 | < h^{p,p}(H^n(X(\C))).
\end{equation}
This type of result was extended to cases with isolated singularites in \cite{Dimca},
where one gets
\begin{equation}\label{POineqSing}
| \chi(X(\R)) -1 | \leq \left\{ \begin{array}{ll} \sum_{0\leq q \leq p} h^{q,q}(H^n_c(X(\C))) 
& n=2p \\[2mm]
\sum_{0\leq q\leq p} h^{q,q}(H^n_c(X(\C))) + h^{p+1,p+1}(H^{p+1}_c(X(\C)) & n=2p+1.
\end{array} \right. 
\end{equation}

However, more generally one does not have a Petrovski\u\i--Ole\u\i nik type inequalities to
estimate the virtual Betti numbers and the Euler characteristic $\chi_c(X)$ of
arbitrary real algebraic varieties in terms of the virtual Hodge polynomials of
the complex variety. For complex varieties the virtual Betti numbers can
be computed in terms of the virtual Hodge polynomial. In fact, one can introduce  
the weight $k$ Euler characteristic given by setting
$$ w^k_j(X(\C)) = \sum_{p+q=j} h^{p,q}(H^k_c(X(\C))), $$
which equal $b_k(X)$ for $j=k$ and zero otherwise in the smooth
projective case, and are otherwise equal to the ranks of the quotients of the
weight filtration $w^k_j(X(\C)) = \dim_\C W^k_j(X)/W^k_{j-1}(X(\C))$ on the cohomology
with compact support $H^k_c(X(\C))$. Then for arbitrary complex algebraic varieties
the virtual Betti numbers are given by (\cite{McCroPar})
$$ \beta_j(X(\C)) =  (-1)^j \sum_k (-1)^k w^k_j(X(\C)). $$
In general one does not have a good way to estimate the virtual Betti numbers
$\beta_k(X(\R))$ of a real algebraic variety, nor their alternating sum $\chi_c(X(\R))$, in
terms of the virtual Betti numbers $\beta_k(X(\C))$.

Although in general one cannot estimate $\chi_c(X(\R))$ in terms of
a suitable Petrovski\u\i--Ole\u\i nik type inequality, we will show that
in certain cases one can compute explicitly both $\chi_c(X(\R))$ and the
virtual Hodge numbers of $X(\C)$ as a consequence of being
able to compute explicitly the class $[X]$ in the Grothendieck ring of
varieties.

We will discuss later how these considerations relate 
to the problem of phase transitions in Potts models. In particular, we will see that,
by working with classes in the Grothendieck ring, we obtain some estimates
on the topological complexity of the set of virtual phase transitions of the Potts model
over certain families of finite graphs $G_n$ approximating some infinite graph $G=\cup_n G_n$.

\subsection{Grothendieck classes of Potts model hypersurfaces}
We also introduce the following notation for the classes in the Grothendieck ring 
of the Potts model hypersurfaces.

\begin{defn}\label{PottsGrClassesDef}
Let $[\cZ_G]$ be the class in the Grothendieck ring $K_0(\cV)$ of the Potts model 
hypersurface \eqref{PottsModHyp}. Also let $\{ \cZ_G \}$ denote the class of
the hypersurface complement,
\begin{equation}\label{hypcomplPotts}
\{ \cZ_G \} =[\A^{\# E(G)+1}\smallsetminus \cZ_G]= \bL^{\# E(G)+1} - [\cZ_G],
\end{equation}
where $\bL=[\A^1]$ is the Lefshetz motive (the class of the affine line).
The classes $[\cZ_{G,q}]$ and $\{ \cZ_{G,q} \}=[\A^{\# E(G)}\smallsetminus \cZ_{G,q}]$
are similarly defined for the hypersurface $\cZ_{G,q}$ of \eqref{qPottsModHyp}.
\end{defn}

As in the case of the graph hypersurfaces of Feynman graphs (see \cite{AluMa3},
\cite{AluMa2}, \cite{BeBro}), we will see that for Potts models it is simpler to
write explicit formulae for the class of hypersurface complement $\{ \cZ_G \}$ than for the
class $[\cZ_G]$ of the hypersurface itself, though the information is clearly 
equivalent due to the simple relation \eqref{hypcomplPotts} between them.

\subsection{The Grothendieck class for fixed $q$ and the fibration condition}

We discuss here the relation between the classes of the hypersurface complement
$\{ \cZ_G \}$ and $\{ \cZ_{G,q} \}$ for the full Potts model hypersurface and for the
one with fixed $q$. We identify a useful condition, according to which the 
the class of $\{ \cZ_G \}$ behaves as one would expect
in the case of a fibration on the locus $q\ne 0,1$. We will later
identify specific classes of graphs we want to work with and check that
they satisfy this condition. We do not address in this paper the question of
how general this condition actually is, nor the question of whether
the variety $\cZ_G$ itself really is a locally trivial  fibration over the locus
$q\ne 0,1$, at least for some specific families of graphs. 

One can see directly from the poynomial $Z_G(q,t)$ why $q=0$ and $q=1$ should
certainly be special values, for the following reasons. Recall that 
the equation we are dealing with is
\[
Z_G(q,t)=\sum_{A\subseteq E(G)} q^{k(A)} \prod_{a\in A} t_a .
\]

\begin{itemize}
\item For $q=0$ and $G$ nonempty, this is $0$: indeed, $k(A)>0$ for every
subset $A$ of edges in that case. This just says that the hypersurface
$Z_G(q t)=0$ has a component along $q=0$; this component
can be removed (dividing $Z_G(q,t)$ by $q^{k(G)}$) and the residual 
hypersurface may be studied over $q=0$. This is the hypersurface
$Q_G(0,t)=0$ of \eqref{QG0t}, which is the dual $\cY_G$ of the graph 
hypersurface $\cX_G$ as in Definition \ref{GraphHypDef}. 
This falls back on the case investigated in \cite{AluMa2}, \cite{AluMa3}.

\item For $q=1$, the polynomial becomes
\begin{equation}\label{ZG1t}
Z_G(1,t)= \sum_{A\subseteq E} \, \prod_{a\in A} t_a
=\prod_{e\in E} (1+t_e).
\end{equation}
This is a union of normal crossing divisors, and in fact it consists of
essentially $n$ coordinate hyperplanes in $\Abb^n$. 
The complement is the set of $n$-tuples $(t_1,\dots,t_n)$ ($n=\#E(G)$) 
with each $t_i+1\ne 0$, a copy of the $n$-torus. Thus the class of its 
complement is $\Tbb^n$.
\end{itemize}

The condition that the class $\{ \cZ_G \}$ behaves as in the case of a 
fibration over the locus $q\ne 0,1$ can then
be formulated as the condition that
\begin{itemize}
\item The class $\{\cZ_{G,q}\}$ is independent of $q$ for $q\ne 0,1$; this
class will be denoted $\{\cZ_{G, q\neq 0,1}\}$;
\item The following holds:
\begin{equation}\label{ZGfibrationq}
\{\cZ_G\} = (\Tbb-1) \{\cZ_{G, q\neq 0,1}\}  + \Tbb^{\#E(G)} .
\end{equation}
\end{itemize}
This accounts for the fact that the complement of $Z_G=0$ is 
contained in $q\ne 0$, has a torus fiber over $q=1$, and (heuristically)
fibers over $q\ne 0,1$, a copy of $\Tbb-1$, with constant fiber class.

In particular, a necessary condition for \eqref{ZGfibrationq} is that
 $(\Tbb-1)$ divides $\{\cZ_G\}-\Tbb^{\# E(G)}$ in the Grothendieck
ring, so any counterexample to this property would give examples
where the fibration condition \eqref{ZGfibrationq} does not hold.

\section{Deletion--contraction for classes in the Grothendieck ring}

In \cite{AluMa3} it was shown that, in the case of the graph hypersurface
complements $\A^{\# E(G)}\smallsetminus \cX_G$, the classes in the
Grothendieck ring satisfy an algebro geometric analog of a deletion--contraction
relation. More precisely, it was proved in \cite{AluMa3} that, for a graph $G$ with $n=\# E(G)$
and for a given edge $e\in E(G)$, the classes of the varieties
$\cX_G$, $\cX_{G/e}$ and $\cX_{G\smallsetminus e}$ are related by
\begin{equation}\label{delconXG}
 [\A^n \smallsetminus \cX_G] = \bL\, [ \A^{n-1}\smallsetminus (\cX_{G\smallsetminus e} \cap
\cX_{G/e})] - [\A^{n-1} \smallsetminus \cX_{G\smallsetminus e} ] ,
\end{equation}
when $e$ is neither a bridge nor a looping edge, and
$$ [\A^n \smallsetminus \cX_G] =\bL\, [\A^{n-1}\smallsetminus
\cX_{G/e}] = \bL\, [\A^{n-1}\smallsetminus \cX_{G\smallsetminus e}] $$
when $e$ is a bridge and
$$ [\A^n \smallsetminus \cX_G] =(\bL-1) [\A^{n-1}\smallsetminus
\cX_{G/e}] =(\bL-1) [\A^{n-1}\smallsetminus \cX_{G\smallsetminus e}] $$
when $e$ is a looping edge, where $\bL=[\A^1]$ is the Lefschetz motive, as above.

Notice how \eqref{delconXG} is not a combinatorial deletion--contraction formula:
indeed, the term involving the intersection $\cX_{G\smallsetminus e} \cap
\cX_{G/e}$ of the hypersurfaces of the deletion and the contraction is in
general difficult to control explicitly, even if one has an explicit formula for
the classes of the deletion and the contraction separately.  However, it
was also shown in \cite{AluMa3} that, for certain families of graphs, such
as chains of polygons, one obtains recursive relations in the 
Grothendieck ring, where the ``problematic"
term in the deletion--contraction formula cancels out and one obtains an
explicit generating function for the classes of the varieties associated to 
the family of graphs. The result provides a way to control how the class
in the Grothendieck ring grows in complexity when the graph is enlarged
through some simple operations, such as doubling edges or splitting edges.
In the setting of quantum field theory the families of graphs obtained
through such simple operations, like the chains of polygons, are typically 
not complex enough to capture interesting behaviors of the associated periods, 
but we will argue here that analogous operations performed in the 
setting of Potts models already gives rise to interesting non-trivial cases.

\subsection{Algebro-geometric deletion--contraction for Potts model hypersurfaces}

We prove here an analog of the deletion-contraction formula \eqref{delconXG}
for the classes in the Grothendieck ring of the Potts model hypersurfaces. We first
analyze the case of the full $\cZ_G$ and then the case of $\cZ_{G,q}$ with fixed $q$
and of the tangent cone $\cV_G$ at zero. 

We now state our main result on the deletion--contraction properties for 
Potts model hypersurfaces. Notice that, in the following statement,
there is no distinction between the case of bridges or looping edges and all the
other edges, just as in the combinatorial deletion--contraction relation for the
multivariate Tutte polynomials.

\begin{thm}\label{delconZGthm}
Let $G$ be a finite graph and $e$ an edge of $G$. Then the class $\{ \cZ_G \}$ 
of \eqref{hypcomplPotts} satisfies
\begin{equation}\label{delconZG}
 \{ \cZ_G \} = \bL \{ \cZ_{G/e} \cap \cZ_{G\smallsetminus e} \} - \{ \cZ_{G/e} \} .
\end{equation}
\end{thm}

\proof
The result follows from the combinatorial deletion--contraction relation for
the multivariate Tutte polynomials
\[
Z_G(q,t)=Z_{G\smallsetminus e}(q,\hat t^{(e)})
+ t_e\, Z_{G/e}(q,\hat t^{(e)}),
\]
where $\hat t^{(e)}$ is the set of variables $t=(t_{e'})_{e'\in E(G)}$ with the
variable $t_e$ omitted. We then check the various cases. 

$\bullet$ If $Z_{G/e}(q,\hat t^{(e)})\ne 0$, then $Z_G(q,t)$ is
guaranteed to be $\ne 0$ if $t_e$ does not equal $-Z_{G\smallsetminus
e}(q,\hat t^{(e)})/ Z_{G/e}(q,\hat t^{(e)})$. This accounts
for a $\bG_m$ worth of $t_e$'s for each such $(q,\hat t^{(e)})$,
contributing a class 
\[
(\Lbb-1) \{ \cZ_{G/e} \} .
\]

$\bullet$ If $Z_{G/e}(q,\hat t^{(e)})= 0$, then $Z_G(q,t)\ne 0$ if and only 
if $Z_{G\smallsetminus e}(q,\hat t^{(e)})\ne
0$. This accounts for an $\A^1$ worth of $t_e$'s for each $(q,\hat t^{(e)})$ 
such that $Z_{G/e}(q,\hat t^{(e)})= 0$ and
$Z_{G\smallsetminus e}(q,\hat t^{(e)})\ne 0$. This contributes a
class
\[
\Lbb \cdot [\cZ_{G/e} \smallsetminus (\cZ_{G/e} \cap \cZ_{G\smallsetminus
e})] .
\]
Note that
\[
[\cZ_{G/e}]-[\cZ_{G/e} \cap \cZ_{G\smallsetminus e}] = \Lbb^{\#E(G)}-[\cZ_{G/e}
  \cap \cZ_{G\smallsetminus e}] -\Lbb^{\#E(G)}+[\cZ_{G/e}] \] \[ =\{ \cZ_{G/e} \cap
\cZ_{G\smallsetminus e} \} - \{ \cZ_{G/e} \} .
\]

Thus, the two contributions add up to
\[
\{\cZ_G\}=(\Lbb-1) \{ \cZ_{G/e} + \Lbb\left(
\{\cZ_{G/e}\cap \cZ_{G\smallsetminus e}\}-\{\cZ_{G/e}\}\right)
= \Lbb \{\cZ_{G/e}\cap \cZ_{G\smallsetminus e}\}- \{\cZ_{G/e}\} ,
\]
as claimed.
\endproof

The following properties of the classes of Potts model hypersurfaces are
simple consequences of the definitions, or follow easily from Theorem~\ref{delconZGthm}.

\begin{cor}\label{listZGprop}
The classes $\{ \cZ_G \} \in K_0(\cV)$ of \eqref{hypcomplPotts} satisfy the
following properties:
\begin{enumerate}
\item If $G$ consists of a single vertex and no edges, then  $\{ \cZ_G \} =\bL-1$.
\item If $G$ consists of a single edge, with either one or two vertices,
then $\{ \cZ_G \} = (\bL-1)^2$.
\item If a graph $G'$ is the union $G'=G_1 \cup_v G_2$ of two graphs joined at a vertex $v$, and
$G''$ denotes the disjoint union of the same two graphs, then $\{ \cZ_{G'} \} = \{ \cZ_{G''} \}$.
\item If $G$ is obtained
by joining $G_1$ and $G_2$ with a single edge from a vertex of $G_1$
to a vertex of $G_2$, then $\{\cZ_G\} = (\Lbb-1) \{\cZ_{G'}\}$, with $G'=G_1 \cup_v G_2$ as above.
\item If $\overline G$ is obtained from a graph $G$ by appending a single (looping or
otherwise) edge to a vertex, then $\{\cZ_{\overline G}\} = (\Lbb-1) \{\cZ_G\}$.
\end{enumerate}
\end{cor}

\proof (1) If $G$ consists of a single vertex, then $Z_G=q$ defines a point in the
affine line $\Abb^1$ with coordinate $q$.

(2) If $G$ is a single edge 
joining two distinct vertices, then $Z_G(q,t)=qt+q^2$ and therefore
$\{\cZ_G\}=\Lbb^2-2\Lbb+1=(\Lbb-1)^2$. 

If $G$ is a single looping edge, then
$\cZ_G(q,t)=qt+q$ and again $\{\cZ_G\}=(\Lbb-1)^2$. 

(3) If $G'$ consists of 
two graphs $G_1$ and $G_2$ joined at a vertex, then $Z_{G'}(q,t)=
\frac 1q Z_{G_1} Z_{G_2}$. If $G''$ consists of the disjoint
union of $G_1$ and $G_2$, then $Z_{G''}(q,t)=Z_{G_1} Z_{G_2}$. It follows (if
none of the graphs is empty) that $\{\cZ_{G'}\}=\{\cZ_{G''}\}$, and that
$\{\cZ_{G'}\cap \cZ_{G''}\} =\{\cZ_{G'}\}$. 

(4) Applying
Theorem~\ref{delconZGthm}, this implies that if $G$ is obtained
by joining $G_1$ and $G_2$ with a single edge from a vertex of $G_1$
to a vertex of $G_2$, then $\{\cZ_G\} = (\Lbb-1) \{\cZ_{G'}\}$.

(5) Let $G'$ be the disjoint union of $G$ and of a single vertex. Then
$Z_{G'} =qZ_G$, and in particular $\{Z_{G'}\}=\{Z_G\}$, and
$Z_{G'}\cap Z_G=Z_G$.  Applying Theorem~\ref{delconZGthm}, we see
that if $\overline G$ is obtained by appending a single non-looping
edge to a vertex of $G$, then $\{\cZ_{\overline G}\} = (\Lbb-1) \{\cZ_G\}$.
The same conclusion is reached if $\overline G$ is obtained from
$G$ by adding a looping edge.
\endproof

Notice that, unlike what happens with the graph hypersurfaces of Feynman 
diagrams (see \cite{AluMa2}), in the case of a disjoint union of graphs the 
class of $\{\cZ_{G_1\cup G_2}\}$ is not the product of the classes of the 
hypersurface complements of the two graphs, since here the polynomials 
$Z_{G_1}$ and $Z_{G_2}$ have the same variable $q$ in common.
However, if the graph $G=G_1\cup G_2$ is given by the disjoint union of 
two graphs $G_1$ and $G_2$, and all the graphs involved satisfy the fibration 
condition \eqref{ZGfibrationq}, then we can obtain an explicit formula for the 
class $\{ \cZ_{G_1 \cup G_2} \}$.

\begin{cor}\label{disjunfibrcor}
Let $G=G_1\cup G_2$ be the disjoint union of two finite graphs $G_1$ and $G_2$,
and assume that the classes $\{ \cZ_G \}$ and  $\{ \cZ_{G_i} \}$ satisfy the fibration condition 
\eqref{ZGfibrationq}. Then the class $\{ \cZ_G \}$ can be expressed
explicitly in terms of the classes $\{ \cZ_{G_1} \}$ and $\{ \cZ_{G_2} \}$ by
\begin{equation}\label{disjunfibr}
\{ \cZ_{G_1\cup G_2} \}=
\frac{\{\cZ_{G_1}\}\cdot \{\cZ_{G_2}\}-\Tbb^{\# E(G_1)}\cdot \{\cZ_{G_2}\}
-\Tbb^{\#E(G_2)}\cdot \{\cZ_{G_1}\}+\Tbb^{\#E(G_1)+\#E(G_2)+1}}{\Tbb-1}.
\end{equation}
\end{cor}

\proof For  $q$ fixed, the remaining variables are indeed
distinct for disjoint $G_1$, $G_2$. Thus, the classes satisfy
\begin{equation}\label{Feyruleq}
\{\cZ_{G_1\cup G_2,q}\} 
=\{\cZ_{G_1,q}\} \cdot \{\cZ_{G_2,q}\} .
\end{equation}
(This holds for all $q$, including the special values $q=0$ and $q=1$.) 
If the classes $\{ \cZ_{G_i,q} \}$ are independent of 
$q\ne 0,1$ and the formula \eqref{ZGfibrationq} holds,
then we get
\begin{align*}
\{\cZ_{G_1\cup G_2}\}
&=(\Tbb-1) \{\cZ_{G_1\cup G_2, q\neq 0,1}\} +\Tbb^{\# E(G_1)+\# E(G_2)}\\
&=(\Tbb-1) (\{\cZ_{G_1, q\neq 0,1}\} \cdot \{\cZ_{G_2, q\neq 0,1}\}) 
+\Tbb^{\#E(G_1)+\#E(G_2)}\\
&=(\Tbb-1) \left(\frac{\{\cZ_{G_1}\}-\Tbb^{\#E(G_1)}}{\Tbb-1} \cdot 
\frac{\{\cZ_{G_2}\}-\Tbb^{\#E(G_2)}}{\Tbb-1} \right) 
+\Tbb^{\#E(G_1)+\#E(G_2)}\\
&=\frac{\{\cZ_{G_1}\}\cdot \{\cZ_{G_2}\}-\Tbb^{\#E(G_1)}\cdot \{\cZ_{G_2}\}
-\Tbb^{\# E(G_2)}\cdot \{\cZ_{G_1}\}+\Tbb^{\#E(G_1)+\#E(G_2)+1}}{\Tbb-1}
\end{align*}
\endproof

\subsection{Deletion--contraction for fixed $q$}

Deletion--contraction works exactly as in the case of the full Potts model
hypersurface $\cZ_G$. 

\begin{prop}\label{deleconPottsq}
For a finite graph $G$, the class $\{ \cZ_{G,q} \}$ for fixed $q$ satisfies
\begin{equation}\label{delconZGq}
\{\cZ_{G,q} \}=(\Tbb+1) \{\cZ_{G/e,q}\cap \cZ_{G\smallsetminus e,q}\}- \{\cZ_{G/e,q}\} .
\end{equation}
\end{prop}

The argument is identical to the one used in the proof of Theorem \ref{delconZGthm}.
(For $q=0$, all classes equal~$0$.) We then also have the analog of 
Corollary~\ref{listZGprop}.

\begin{cor}\label{listZGqprop}
For $q\ne 0$, the classes $\{ \cZ_{G,q} \}$ satisfy the
following properties:
\begin{enumerate}
\item For $G$ a single vertex, $\{ \cZ_{G,q} \}=1$.
\item For $G$ a single edge joining either one or two vertices, $\{\cZ_{G,q}\}=\Tbb$.
\item If $G'$ consists of two graphs $G_1$ and $G_2$ joined at a vertex, or 
disjoint, then $\{\cZ_{G',q}\}=\{\cZ_{G_1,q}\} \cdot \{Z_{G_2,q}\}$.
\item If $G$ is obtained by joining $G_1$ and $G_2$ with a single edge from a
vertex of $G_1$ to a vertex of $G_2$, then
$\{\cZ_{G,q}\} = \Tbb \{\cZ_{G',q}\}$.
\item Appending a single edge to a vertex of $G$ multiplies the class 
$\{\cZ_{G,q}\}$ by $\Tbb$.
\end{enumerate}
\end{cor}

\proof (1) For $G$ a single vertex, $Z_G(q,t)=q\ne 0$; thus 
$\cZ_{G,q}=\emptyset\subset \Abb^0$, and $\{ \cZ_{G,q} \}=1$.

(2) For $G$ a single edge joining two distinct vertices, $Z_G(q,t)=qt+q^2$ defines
a point in $\Abb^1$; thus $\{\cZ_{G,q}\}=\Tbb$.
For $G$ a single looping edge, $Z_G(q,t)=qt+q$ and again $\{\cZ_{G,q}\}=\Tbb$.

(3)  This case is simpler than the case for $\{ \cZ_G \}$ with indeterminate $q$.
It follows, as in \eqref{Feyruleq}, from the fact that the polynomials $Z_{G_1}(q,t)$
and $Z_{G_2}(q,t)$ have none of the variables other than the fixed $q$ in common.

(4)  This again follows from the simpler formula for unions as in (3), 
and from the computation for a single edge. 
The relation is the same as in the indeterminate-$q$ case.

(5) This is again the same formula as in the free case,
which here follows from (3) and the case of a single edge.
\endproof

\subsection{Deletion--contraction for the tangent cone}  
  
In the case of the tangent cone $\cV_G$ at zero of the variety $\cZ_G$, which, as
we have seen interpolates between the Potts model and the graph hypersurfaces
considered in the quantum field theory context, it is convenient to introduce the
following notation for the classes in the Grothendieck ring.

We still denote by $[\cV_G]$ the class of $\cV_G$ and by $\{ \cV_G \}  =[\bA^{\#E(G)+1}\smallsetminus \cV_G ]$ the class of the complement. We also use the notation 
$\cY_G$ for the graph hypersurface given by the intersection of the component
$\cW_G=\{ Q_G(q,t)=0\}$ with the hyperplane $H=\{ q=0 \}$. 
Thus, $\cY_G$ is the locus of zeros of $Q_G(0,t)$, where 
$P_G(q,t)=q^{k(G)} Q_G(q,t)$. This gives, at the level of the classes
\begin{equation}\label{coneVWY}
 [\cV_G]= [\cW_G] + \bL^{\# E(G)} - [\cY_G], 
\end{equation}
where $\bL^{\# E(G)}=[H]$ and $\cV_G=\cW_G \cup H$ with $\cY_G =\cW_G \cap H$.
Thus, we obtain the following.

\begin{lem}\label{complVG}
The class $\{ \cV_G \}$ of the complement of the tangent cone $\cV_G$ in
$\bA^{\#E(G)+1}$ is given by the class $\{ \cV_G \}=\{ \cW_G \} - \{ \cY_G \}$.
\end{lem}

\proof Using \eqref{coneVWY} and taking complements, we obtain
$$ \bL^{\#E(G)+1} - [ \cV_G ] = \bL^{\#E(G)+1} - [\cW_G] - \bL^{\# E(G)} + [\cY_G]=
\{ \cW_G \}-\{\cY_G\}. $$
Alternately, one may simply observe that the complement of $\cW_G$ is the 
disjoint union of $\cV_G$ and $cY_G$.
\endproof

We can then see directly that the class $\{ \cV_G \}$ satisfies the following simple
properties.

\begin{lem}\label{propVGlist}
The class $\{ \cV_G \}$ satisfies:
\begin{enumerate}
\item If $G'$ is obtained by attaching a looping edge to a vertex of $G$, then
$$\{ \cV_{G'} \}=(\bT+1) \{ \cV_G \}.$$
\item If $G'$ is obtained by appending a non-looping edge to a vertex of $G$, then
$$\{ \cV_{G'} \}=\bT \{ \cV_G \}.$$
\item More generally, if $G$ is obtained by connecting two disjoint graphs by a 
bridge~$e$, then $\{\cV_G\}=\Tbb \{\cV_{G/e}\}$.
\item If $G'$ is obtained by attaching an edge parallel to one of the edges of $G$,
then $$\{ \cV_{G'} \}=(\bT+1) \{ \cV_G \}.$$
\end{enumerate}
\end{lem}

\proof (1) The polynomial $P_G$ is only counting forests and loops are excluded
from forests, so the polynomial $P_{G'}$ equals
$P_G$, but it is viewed in one dimension higher. This simply multiplies
everything by $\Lbb=\Tbb+1$.

(2) Attaching an unconnected edge $e$ multiplies $P_G$ by $(q+t_e)$.
The condition $(t+q)P_G\ne 0$ implies $P_G\ne 0$ and
$t+q\ne 0$. This says that the complement to $P_{G'}=0$ fibers over
the complement to $P_G=0$, with $\Tbb$ fibers.

(3) The same argument proves this assertion.

(4) Let $e$ be the edge of $G$ that we are doubling,
and call $f$ the new parallel edge. Then $P_{G'}$ is obtained from 
$P_G$ by replacing $t_e$ by $t_e+t_f$. This operation
amounts to taking a cylinder, hence multiplying everything by 
$\Lbb=\Tbb+1$.
\endproof

We then look at the deletion contraction formula. In the cases of $\{ \cZ_G \}$
and $\{ \cZ_{G,q} \}$ considered above, we did not have to make a special
case for bridges and looping edges because the combinatorial deletion--contraction
formula for the multivariate Tutte polynomial does not make such a distinction.
However, in the case of the tangent cone, as in the case of the graph hypersurfaces
of Feynman diagrams of \cite{AluMa3}, we need to take these two special cases
into account separately. 

When the edge $e$ is neither a bridge nor a looping edge, 
the deletion-contraction formula for the multivariate Tutte polynomial
specializes to one for $P_G$:
\begin{equation}\label{PGdelcon}
P_G(q,t)=P_{G\smallsetminus e}(q,\hat t^{(e)})
+ t_e\, P_{G/e}(q,\hat t^{(e)}).
\end{equation}
In this case the numbers 
of connected components of $G$, $G\smallsetminus e$, $G/e$ are all 
equal. Therefore, the same formula holds for~$Q_G$. It also holds once 
$q$ is set to $0$ in the latter. 

We say that an edge is a {\em regular edge} if it is neither a bridge
nor a looping edge.

\begin{prop}\label{deleconPottstc}
Assume $e$ is a regular edge of $G$. Then
\begin{gather*}
\{ \cW_G \}=\Lbb \cdot [\Abb^{\#E(G)}-(\cW_{G\smallsetminus e}\cap \cW_{G/e})]
-\{ \cW_{G/e} \}\quad,\\
\{ \cY_G\} =\Lbb \cdot [\Abb^{\#E(G)-1}
-(\cY_{G\smallsetminus e}\cap \cY_{G/e})]
-\{ \cY_{G/e}\}.
\end{gather*}
and the formula for $\{ \cV_G \}$ is then the difference of these.
\end{prop}

\proof 
The proof is entirely analogous to the one given for $\{ \cZ_G \}$, 
applied to the polynomials $Q_G$ and $Q_G|_{q=0}$. 
\endproof

The cases of bridges and looping edges are already dealt with in 
Lemma~\ref{propVGlist}; we repeat them here for clarity:

\begin{prop}\label{deleconPottstcbloop} 
If $e$ is a looping edge of $G$, then 
$$ \{\cV_G\}=(\Tbb+1) \{\cV_{G\smallsetminus e}\} .$$
If $e$ is a bridge, then 
$$ \{\cV_G\}=\Tbb\,\{\cV_{G/e}\} .$$
\end{prop} 

\section{Edge splitting}

We now use the deletion--contraction formula of Theorem~\ref{delconZGthm}
to describe the effect of splitting an edge in a graph.

\begin{defn}\label{splitndef}
Given a finite graph $G$ and an edge $e\in E(G)$, let $\bc 0G$ denote the
contraction $G/e$; let $\bc 1G=G$; and more generally let
$\bc kG$ be the graph obtained by replacing the edge $e$ with a chain of $k\geq 2$
edges. 
\end{defn} 

\subsection{Splitting an edge}

We have the following formula for the class of the hypersurface
complement of the graph $\bc 2G$.
Here and in the following, we denote by $V(f_1,\dots,f_k)$ the zero locus
$\{f_1=\cdots=f_k=0\}$ of the ideal generated by $(f_1,\dots,f_k)$.

\begin{thm}\label{resplitPotts}
For $G$ a finite graph and $\bc 2G$ the graph obtained by splitting an edge $e$
in $G$, the class $\{Z_{\bc 2G}\}$ satisfies
\begin{equation}\label{Z2Gformula}
\{\cZ_{\bc 2G}\} = 
(\Tbb-2) \{\cZ_{\bc 1G}\} + (\Tbb-1) \{ \cZ_{\bc 0G}\} + 
(\Tbb+1) \left(\{\cZ_{G\smallsetminus e}\}+\{A^e_G\}\right),
\end{equation}
where $A^e_G=V(q+t_e,Z_{G\smallsetminus e}-q Z_{G/e})$.
\end{thm}

The term $\{A^e_G\}$ appears to be difficult to evaluate geometrically.
Equation~\ref{Z2Gformula} gives a relation between this term and the
terms $\{Z_{\bc mG}\}$, which will allow us to obtain recursive formulas
for these classes which are independent of $\{A^e_G\}$.

\proof
First observe that the effect of attaching an edge $e$ to a vertex of a graph $G$ 
is to multiply the polynomial $Z_G(q,t)$ by $(t_e+q)$ in the
case of a non-looping edge, and by $(t_e+1)$ in the case of a looping
edge.  This has the effect, in both cases, of multiplying the class $\{ \cZ_G \}$ 
by the factor $\Lbb-1$ as seen in Corollary \ref{listZGprop}, (5).

Applying Theorem~\ref{delconZGthm} to the graph $\bc 2G$ requires handling
the contraction, $\bc 1G$ in this case, and the deletion, which is obtained
from $G\smallsetminus e$ by adding a non-looping edge. In terms of 
equations, these are given respectively by the vanishing of $Z_G(q,t)$ and of
$Z_{G\smallsetminus e}(q,\hat t^{(e)})\cdot (q+t_e)$. The most interesting term 
is, of course, the class of the intersection of these two hypersurfaces, with ideal
\begin{equation}\label{int2Gideal}
(Z_{G\smallsetminus e}(q,\hat t^{(e)}) (q+t_e), Z_G(q,t)) .
\end{equation}
This ideal defines a subscheme of $\Abb^{\# E(G)+1}$, while $\cZ_{\bc 2G}$
lives in $\Abb^{\#E(G)+2}$. Deletion-contraction applied to $G$ gives
\[
Z_G=Z_{G\smallsetminus e}+t_e Z_{G/e},
\]
therefore the zero locus $V(Z_{G\smallsetminus e} (q+t_e), Z_G)$ equals
\begin{align*}
V(Z_{G\smallsetminus e} (q+t_e), & Z_{G\smallsetminus e}+t_e Z_{G/e})
=V(Z_{G\smallsetminus e},t_e Z_{G/e})\cup 
V(q+t_e,Z_{G\smallsetminus e}+t_e Z_{G/e}) \\
&=V(Z_{G\smallsetminus e},t_e)\cup V(Z_{G\smallsetminus e},Z_{G/e})
\cup V(q+t_e,Z_{G\smallsetminus e}-q Z_{G/e}) .
\end{align*}
To apply inclusion-exclusion, we need the double and triple intersections of
these components:
\begin{align*}
&V(Z_{G\smallsetminus e},t_e,Z_{G\smallsetminus e},Z_{G/e})
=V(t_e, Z_{G\smallsetminus e},Z_{G/e}) \\
&V(Z_{G\smallsetminus e},t_e,q+t_e,Z_{G\smallsetminus e}-q Z_{G/e})
=V(q,t_e,Z_{G\smallsetminus e})=V(q,t_e) \\
&V(Z_{G\smallsetminus e},Z_{G/e},q+t_e,Z_{G\smallsetminus e}-q Z_{G/e})
=V(q+t_e,Z_{G\smallsetminus e},Z_{G/e})
\end{align*}
and
\[
V(Z_{G\smallsetminus e},t_e,Z_{G\smallsetminus e},Z_{G/e},
q+t_e,Z_{G\smallsetminus e}-q Z_{G/e})
=V(q,t_e,Z_{G\smallsetminus e},Z_{G/e})=V(q,t_e) ,
\]
where we used the fact that $Z_{G\smallsetminus e}$ and $Z_{G/e}$ 
are multiples of $q$. 
This implies that the triple intersection
is in fact a double intersection, causing a useful cancellation at the level of Grothendieck
classes:
\begin{multline*}
[V(Z_{G\smallsetminus e} (q+t_e), Z_G)]
=
[V(Z_{G\smallsetminus e},t_e)]+[V(Z_{G\smallsetminus e},Z_{G/e})]
+[ V(q+t_e,Z_{G\smallsetminus e}-q Z_{G/e})] \\
-[V(t_e, Z_{G\smallsetminus e},Z_{G/e})] - [V(q+t_e,Z_{G\smallsetminus e},Z_{G/e})] .
\end{multline*}

All but one of the classes on the right-hand side have a clear interpretation:
\[
[V(Z_{G\smallsetminus e},t_e)] = [\cZ_{G\smallsetminus e}] ,
\]
where we view $\cZ_{G\smallsetminus e}$ as a hypersurface of $\Abb^{\# E(G)}$;
\[
[V(Z_{G\smallsetminus e},Z_{G/e})] = \Lbb\cdot [\cZ_{G\smallsetminus e}\cap \cZ_{G/e}] ,
\]
where the intersection is again viewed in $\Abb^{\# E(G)}$ and the factor of $\Lbb$ is due to
the free variable~$t_e$;
\[
[V(t_e, Z_{G\smallsetminus e},Z_{G/e})] = [V(q+t_e,Z_{G\smallsetminus e},Z_{G/e})]
=[\cZ_{G\smallsetminus e}\cap \cZ_{G/e}],
\]
still in $\Abb^{\# E(G)}$: indeed, $t_e$ and $q+t_e$ may be used to eliminate $t_e$ in both
cases, and both ideals define the same locus in $\Abb^{\# E(G)}$ (with variables 
$(q,\hat t^{(e)})$) after this projection.

The remaining term is 
\[
[V(q+t_e,Z_{G\smallsetminus e}-q Z_{G/e})]  = [ V(Z_{G\smallsetminus e}-q Z_{G/e})]  ,
\]
where again we use $q+t_e=0$ to eliminate $t_e$, and view the right-hand-side as
the class of a locus in $\Abb^{\# E(G)}$, with variables $(q,\hat t^{(e)})$.

Let $A^e_G$ denote the locus determined by this ideal, as a subset of $\Abb^{\# E(G)}$. 
We have obtained that
\[
[V(Z_{G\smallsetminus e} (q+t_e), Z_G)]
=(\Lbb-2)\cdot [Z_{G\smallsetminus e}\cap Z_{G/e}]
+[Z_{G\smallsetminus e}]+[ A^e_G] .
\]
This can be equivalently stated in terms of classes of hypersurface complements as
\begin{equation}\label{int2G}
\{V(Z_{G\smallsetminus e}(q+t_e),Z_G)\}
=(\Lbb-2)\cdot \{Z_{G\smallsetminus e}\cap Z_{G/e}\} + \{Z_{G\smallsetminus e}\}
+\{A^e_G\} .
\end{equation}
Indeed, we have
\[
\{V(Z_{G\smallsetminus e}(q+t_e),Z_G)\} = \Lbb^{|E|+1}-
[V(Z_{G\smallsetminus e}(q+t_e),Z_G)]\quad,
\]
while the complements of the other loci are taken in $\Abb^{|E|}$. Thus,
\begin{multline*}
\{V(Z_{G\smallsetminus e}(q+t_e),Z_G)\} = 
\Lbb^{\# E(G)+1}-\left((\Lbb-2)\cdot [\cZ_{G\smallsetminus e}\cap \cZ_{G/e}]
+[\cZ_{G\smallsetminus e}]+[ A^e_G]\right) \\
= \Lbb (\Lbb^{\# E(G)}-[\cZ_{G\smallsetminus e}\cap \cZ_{G/e}])
-2 (\Lbb^{\# E(G)}-[\cZ_{G\smallsetminus e}\cap \cZ_{G/e}])+2 \Lbb^{\# E(G)} -
([\cZ_{G\smallsetminus e}]+[ A^e_G])
\end{multline*}
with the stated result \eqref{int2G}.

Thus, we have obtained in this way an explicit calculation of the intersection term 
needed to apply Theorem~\ref{delconZGthm} to the graph $\bc 2G$ 
obtained by splitting an edge of $G$ into two. We obtain
\begin{equation}\label{Z2GYeG}
\{\cZ_{\bc 2G}\} = \Lbb\left( (\Lbb-2)\{\cZ_{G\smallsetminus e}\cap \cZ_{G/e}\} 
+ \{\cZ_{G\smallsetminus e}\}+\{A^e_G\} \right)-\{\cZ_{\bc 1G}\}\quad.
\end{equation}

We then apply again Theorem~\ref{delconZGthm} to $G$ to provide
an alternative expression for the intersection term 
$\{\cZ_{G\smallsetminus e}\cap \cZ_{G/e}\}$ and we obtain
\begin{equation}\label{Z2Gintterm}
\Lbb\cdot \{\cZ_{G\smallsetminus e}\cap \cZ_{G/e}\}
=\{\cZ_{G/e}\}+\{\cZ_G\}= \{\cZ_{\bc 0G}\}+\{ \cZ_{\bc 1G}\} .
\end{equation}
\endproof

The locus $A^e_G$ determined by the ideal $(q+t_e,Z_{G\smallsetminus e}-q Z_{G/e})$
has an interpretation in terms of the combinatorics of the graph $G$.
Let $Z'$ denote the sum
\begin{equation}\label{Zprimedefeq}
\sum_{A\subseteq E} q^{k(A)} \prod_{a\in A} t_a
\end{equation}
where the sum is restricted to the subgraphs not including $e$ and connecting the
endpoints of $e$; and let $Z''$ denote the same expression, where the sum is restricted
to the subgraphs not including $e$ and {\em not\/} connecting the endpoints of $e$.
Notice that there is a bijection 
between the monomial of $Z_{G/e}$ and the monomials of $Z_{G\smallsetminus e}$;
in fact,
\begin{equation}\label{bijectioneq}
Z_{G\smallsetminus e}=Z' + Z''\quad, \quad Z_{G/e}=Z' + \frac {Z''}q .
\end{equation}
Indeed, the graphs in $Z''$ lose one connected component when $e$ is
contracted. 

\begin{lem}\label{YeGlocus}
The locus $A^e_G$ may be described as 
$V(q+t_e, (1+t_e) Z')$, where the polynomial $(1+t_e) Z'$ is the
sum over all subgraphs of $G$ (including $e$ or not) which connect the
endpoints of $e$ in some way other than through $e$.
\end{lem}

\proof 
By \eqref{bijectioneq},
\[
Z_{G\smallsetminus e}-q Z_{G/e} = (1-q) Z' .
\]
Modulo $q+t_e$ this equals $(1+t_e) Z'$, and this term has the interpretation
detailed in the statement.
\endproof

\subsection{Multiple splittings}
We can now formulate the result for multiple splittings of an edge $e$ in a graph $G$,
that is, for all the classes $\{ \cZ_{\bc mG} \}$.

\begin{lem}\label{resplitcor}
For all $m\ge 1$,  the classes $\{ \cZ_{\bc {m+1}G} \}$ satisfy
\begin{equation}\label{ZmGYeG}
\{\cZ_{\bc {m+1}G}\} = 
(\Tbb-2) \{\cZ_{\bc mG}\} + (\Tbb-1) \{ \cZ_{\bc {m-1}G}\} + 
(\Tbb+1)\Tbb^{m-1} \left(\{\cZ_{G\smallsetminus e}\}+\{A^e_G\}\right).
\end{equation}
\end{lem}

\proof
If a polynomial
is multiplied by $t+q$, where $t$ is a new indeterminate, then the 
effect on the class $\{\cdot\}$ of the variety defined by that polynomial is 
to multiply it by $\Lbb-1=\Tbb=[\bG_m]$
(cf.~Corollary~\ref{listZGprop}, (5)).
Applying this observation to $\cZ_{G\smallsetminus e}$
and $A^e_G$ shows that
\[
\{ \cZ_{\bc mG\smallsetminus e} \}+\{A^e_{\bc mG}\}
=\Tbb^{m-1} (\{\cZ_{G\smallsetminus e}\}+\{A^e_{G}\}) .
\]
where $e$ denotes the last edge added in the splitting $\bc mG$.
Indeed, $Z_{\bc mG\smallsetminus e}$ is obtained from 
$Z_{G\smallsetminus e}$ by attaching a chain of $m-1$ edges to
$G\smallsetminus e$:
\begin{center}
\includegraphics[scale=.4]{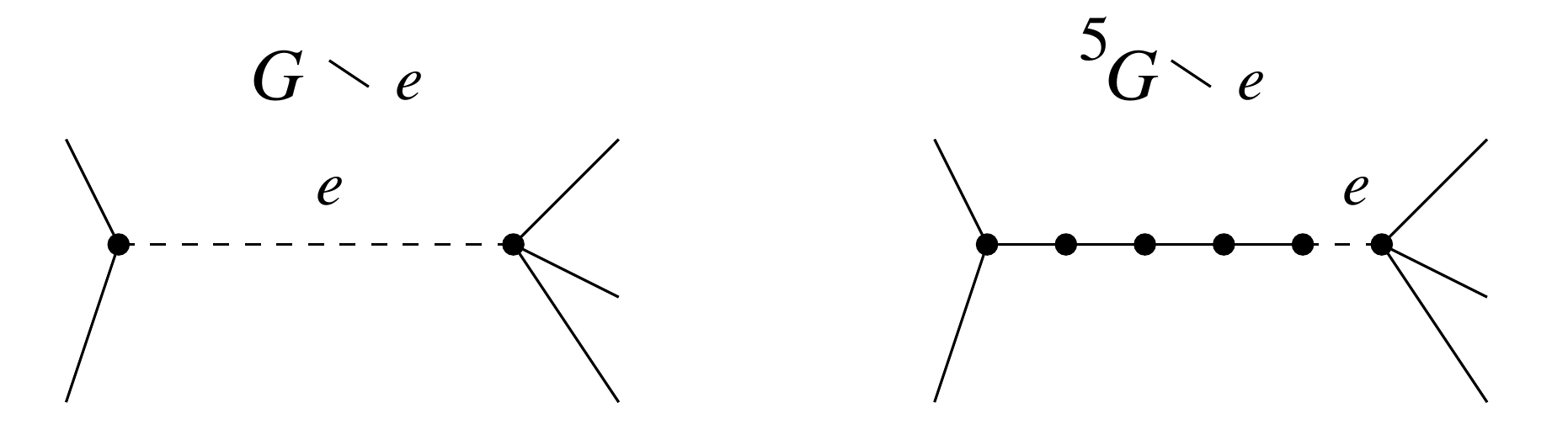}
\end{center}
and this has the effect of multiplying $Z_{G\smallsetminus e}$ by a term
$\prod_{i=1}^{m-1} (t_i+q)$, where $t_1,\dots,t_{m-1}$ are the variables
corresponding to the edges in the chain. Therefore, $\{\cZ_{\bc mG\smallsetminus
e}\}=\Tbb^{m-1} \{\cZ_{G\smallsetminus e}\}$. The effect on $A^e_G$ is precisely 
the same. Indeed, recall that, by Lemma \ref{YeGlocus},
the equation for $A^e_G$
in the hyperplane $t_e=-q$
 is $(1-q)Z'$, where 
$Z'$ is the sum corresponding to subgraphs of $G\smallsetminus e$ which 
connect the endpoints of $e$. From this description it is clear that the 
equation for $A^e_{\bc mG}$
in $\{t_e=-q\}$ 
is $(1-q)Z' \prod_{i=1}^{m-1} (t_i+q)$, and
again it follows that $\{A^e_{\bc mG}\}=\Tbb^{m-1} \{A^e_G\}$.
\endproof

Combining Theorem~\ref{resplitPotts} and Lemma~\ref{resplitcor} allows us to 
obtain a `combinatorial' expression for the class $\{A^e_G\}$.
We then have the following recursive formula for the classes $\{ \cZ_{\bc mG} \}$.

\begin{thm}\label{newdcnew}
For all $m\ge 0$:
\begin{equation}\label{ZmGrecursion}
\{ \cZ_{\bc {m+3}G}\} = 
(2\Tbb-2) \{ \cZ_{\bc {m+2}G}\} - (\Tbb^2-3\Tbb+1) \{ \cZ_{\bc {m+1}G}\} -
\Tbb (\Tbb-1) \{ \cZ_{\bc mG}\} .
\end{equation}
\end{thm}

\proof Starting with \eqref{ZmGYeG}, Theorem~\ref{resplitPotts} allows 
us to express the term 
$\{\cZ_{G\smallsetminus e}\}+\{A^e_G\}$ in terms of splittings:
\begin{equation}\label{YeGZ210G}
(\Tbb+1) (\{\cZ_{G\smallsetminus e}\}+\{A^e_G\})
=\{\cZ_{\bc 2G}\}-(\Tbb-2) \{\cZ_{\bc 1G}\} -(\Tbb-1) \{\cZ_{\bc 0G}\} .
\end{equation}

Using the case $m=2$ in Lemma~\ref{resplitcor} we then get
\[ \{\cZ_{\bc 3G}\} = 
(\Tbb-2) \{\cZ_{\bc 2G}\} + (\Tbb-1) \{ \cZ_{\bc 1G}\} + 
\Tbb (\{\cZ_{\bc 2G}\}-(\Tbb-2) \{\cZ_{\bc 1G}\} -(\Tbb-1) \{\cZ_{\bc 0G}\}) ,
\]
that is,
\[
\{\cZ_{\bc 3G}\} = 
(2\Tbb-2) \{\cZ_{\bc 2G}\} - (\Tbb^2-3\Tbb+1) \{ \cZ_{\bc 1G}\} -
\Tbb (\Tbb-1) \{ \cZ_{\bc 0G}\} .
\]
Then applying this formula to $\bc {m+1}G$ rather than $\bc 1G$ one obtains the
recursion \eqref{ZmGrecursion}.
\endproof

One can then write a generating function for the classes $\{ \cZ_{\bc mG} \}$,
in a way similar to the corresponding result given in \cite{AluMa3} for the
graph hypersurfaces of Feynman graphs.

\begin{thm}\label{ZmGgenfunc}
The generating function of the classes $\{ \cZ_{\bc mG} \}$ is given by
\begin{multline}\label{ZmGgf}
\sum_{m\ge 0} \{Z_{\bc mG}\} \frac{s^m}{m!}
=\left( e^{(\Tbb-1)s} - (\Tbb-1)\cdot \frac{e^{\Tbb s}-e^{-s}}{\Tbb+1}\right)
\{Z_{\bc 0G}\} \\
+\left( (\Tbb-1)\cdot \frac{e^{(\Tbb-1) s}-e^{-s}}{\Tbb} 
-(\Tbb-2)\cdot \frac{e^{\Tbb s}-e^{-s}}{\Tbb+1}\right)
\{Z_{\bc 1G}\}\\
+\left(- \frac{e^{(\Tbb-1) s}-e^{-s}}{\Tbb}
+\frac{e^{\Tbb s}-e^{-s}}{\Tbb+1}\right)
\{Z_{\bc 2G}\} .
\end{multline}
\end{thm}

\proof Putting
\begin{equation}\label{Ggf}
G_e(s):= \sum_{m\ge 0} \{\cZ_{\bc mG}\}\frac{s^m}{m!}\quad,
\end{equation}
the recursion formula \eqref{ZmGrecursion} translates into the differential equation
\begin{equation}\label{eqdiff}
G'''_e(s)=(2\Tbb-2)\, G''_e(s)-(\Tbb^2-3\Tbb+1)\, 
G'_e(s)-\Tbb(\Tbb-1)\, G_e(s) ,
\end{equation}
with solution
\begin{equation}\label{soleqdiff}
G_e(s)=A\, e^{-s}+B\, e^{\Tbb s}+C\,e^{(\Tbb-1) s} .
\end{equation}
We can impose that this series begins with three undetermined
coefficients, and solving for $A$, $B$, $C$ gives
\begin{align*}
A &=\{\cZ_{\bc 0G}\}
+\frac{\{\cZ_{\bc 2G}\}+\{\cZ_{\bc 1G}\}}{\Tbb}
-\frac{\{\cZ_{\bc 2G}\}+3\{\cZ_{\bc 1G}\}+2\{\cZ_{\bc 0G}\}}{\Tbb+1}\\
B &=-\{\cZ_{\bc 1G}\}-\{\cZ_{\bc 0G}\}
+\frac{\{\cZ_{\bc 2G}\}+3\{\cZ_{\bc 1G}\}+2\{\cZ_{\bc 0G}\}}{\Tbb+1}\\
C &=\{\cZ_{\bc 1G}\}+\{\cZ_{\bc 0G}\}
-\frac{\{\cZ_{\bc 2G}\}+\{\cZ_{\bc 1G}\}}{\Tbb}
\end{align*}
This then gives the form \eqref{ZmGgf} for the generating function $G_e(s)$.
\endproof

One can also write \eqref{ZmGgf} in a form that involves explicitly the
term  $A^e_G$, with slightly simpler form of the coefficients, as follows.

\begin{cor}\label{ZmGgenfuncYeG}
The generating function \eqref{ZmGgf} can be equivalently written as
\begin{multline}\label{ZmGgfYeG}
\sum_{m\ge 0} \{Z_{\bc mG}\} \frac{s^m}{m!}
=e^{-s}\left(
\left(1+\frac{e^{\Tbb s}-1}{\Tbb}\right) \{Z_{\bc 0G}\}
+\frac{e^{\Tbb s}-1}{\Tbb} \{Z_{\bc 1G}\}\right. \\
\left. 
+ \left(e^{(\Tbb+1)s}-e^{\Tbb s}-\frac{e^{\Tbb s}-1}\Tbb\right)
(\{Z_{G\smallsetminus e}\} +\{A^e_G\})\right).
\end{multline}
\end{cor}

\proof This follows directly from \eqref{ZmGgf}
and \eqref{YeGZ210G}.
\endproof

\subsection{Edge splitting for fixed $q$}
The discussion is entirely parallel to the one given above for the class $\{ \cZ_G \}$
with variable $q$. One has the analog of Theorem \ref{resplitPotts} in the case with fixed $q$,
given by the following.

\begin{thm}\label{resplitPottsq}
Let $\bc 2G$ be the graph obtained by splitting an edge $e$ in a graph $G$.
Then the class $\{\cZ_{\bc 2G,q}\}$ satisfies
\begin{equation}\label{Z2Gq2}
\{\cZ_{\bc 2G,q}\} = 
(\Tbb-2) \{\cZ_{\bc 1G,q}\} + (\Tbb-1) \{ \cZ_{\bc 0G,q}\} + 
(\Tbb+1) \left(\{\cZ_{G\smallsetminus e,q}\}+\{A^e_{G,q}\}\right) .
\end{equation}
\end{thm}

\proof 
As in the proof of Theorem~\ref{resplitPotts}, the main point is the computation of
the class of $V(Z_{G\smallsetminus e} (q+t_e), Z_G)$.

\begin{lem}\label{int2Gq}
With $q\ne 0$ fixed, the class of the locus $V(Z_{G\smallsetminus e} (q+t_e), Z_G)$
is given by
\[
\{V(Z_{G\smallsetminus e}(q+t_e),Z_G)\}
=(\Tbb-1)\cdot \{\cZ_{G\smallsetminus e,q}\cap \cZ_{G/e,q}\} + \{\cZ_{G\smallsetminus e,q}\}
+\{A^e_{G,q}\},
\]
where $A^e_{G,q}\subseteq \Abb^{\#E(G)-1}$ is the zero-locus of the
polynomial $Z_{G\smallsetminus e} -q Z_{G/e}=(1-q) Z'$, with fixed $q$, with 
$Z'$ as in \eqref{Zprimedefeq} (with fixed $q$).
\end{lem}

\proof
The argument here parallels closely the computation in the proof of 
Theorem~\ref{resplitPotts}. In that computation (for variable~$q$) there is
a key cancellation of a class $[V(q,t_e)]$ due to inclusion-exclusion; for 
fixed $q\ne 0$ this locus is empty. All other classes admit the same interpretation
for fixed $q$ as for variable $q$.

In particular,
the term $A^e_{G,q}$ is the zero locus of $(1-q) Z'$ with
$Z'$ as in Lemma \ref{YeGlocus}
(and $q$ is now a fixed nonzero number).
Note that this is $0$ when
$q=1$; in this case the last summand in the formula in Lemma~\ref{int2Gq} is~$0$.
In general, arguing as in Lemma \ref{YeGlocus}, one sees that $Z'$ is given by
the sum \eqref{Zprimedefeq},
over the range specified there.
\endproof

Implementing deletion-contraction gives then the same formula as in the case with
variable $q$, so that one obtains the following.

\begin{cor}\label{corZ2Gq}
Let $\bc 2G$ be the graph obtained by splitting an edge $e$ in a graph $G$.
Then the class $\{\cZ_{\bc 2G,q}\}$ satisfies
\begin{equation}\label{Z2Gq}
\{\cZ_{\bc 2G,q}\} = \Lbb\left( (\Lbb-2)\{\cZ_{G\smallsetminus e,q}\cap \cZ_{G/e,q}\} 
+ \{\cZ_{G\smallsetminus e,q}\}+\{A^e_{G,q}\} \right)-\{\cZ_{\bc 1G,q}\} .
\end{equation}
\end{cor}

Note again that the term $\{A^e_{G,q}\}$ would be missing in the case $q=1$.

Next, use Proposition~\ref{deleconPottsq} to get a different expression for
$\{\cZ_{G\smallsetminus e}\cap \cZ_{G/e}\}$, and this gives a perfect parallel
to Theorem~\ref{resplitPotts}, and completes the proof of Theorem \ref{resplitPottsq}.
\endproof

We can now pass to the case of multiple splittings, which again is analogous 
to the case with variable $q$.

\begin{thm}\label{newdcnewq}
Let $\bc mG$ be the graph obtained by multiple splitting on an edge $e$ in $G$.
Then, for all $m\ge 0$, the classes for fixed $q$ satisfy
\[
\{\cZ_{\bc {m+1}G,q}\} = 
(\Tbb-2) \{\cZ_{\bc mG,q}\} + (\Tbb-1) \{ Z_{\bc {m-1}G,q}\} + 
(\Tbb+1)\Tbb^{m-1} \left(\{Z_{G\smallsetminus e,q}\}+\{A^e_{G,q}\}\right) ,
\]
which then gives the recursive relation
\begin{equation}\label{ZGmqrec}
\{\cZ_{\bc {m+3}G,q}\} = 
(2\Tbb-2) \{\cZ_{\bc {m+2}G,q}\} - (\Tbb^2-3\Tbb+1) \{ \cZ_{\bc {m+1}G,q}\} -
\Tbb (\Tbb-1) \{\cZ_{\bc mG,q}\} ,
\end{equation}
\end{thm}

\proof The argument is completely analogous to the case with variable $q$. The
key step is  the relation
\[
\{A^e_{\bc mG,q}\} = \Tbb^{m-1} \{A^e_{G,q}\},
\]
that one can see continues to hold in this case by the same argument used before.
Then \eqref{ZGmqrec} is proved by using Theorem~\ref{resplitPottsq} to solve for the class
$(\{\cZ_{G\smallsetminus e,q}\}+\{A^e_{G,q}\})$.
\endproof

The bottom line is that {\em the same recursion\/} holds for any fixed $q\ne 0$
as for the free $q$ case. What will change are the seeds of this recursion, that is,
the values of $\{\cZ_{\bc mG,q}\}$ for $m=0,1,2$; these will naturally be different
for fixed $q$. 

Nothing in the proof of the recursion excludes the case $q=1$, and
indeed $\{\cZ_{\bc mG, q=1}\}=\Tbb^{\#E(G)+m-1}$ is a solution of the recursion
\[
(2\Tbb-2) \Tbb^{\#E(G)+m+1} - (\Tbb^2-3\Tbb+1) \Tbb^{\#E(G)+m} -
\Tbb (\Tbb-1) \Tbb^{\#E(G)+m-1}
=\Tbb^{\#E(G)+m+2}.
\]

With respect to the question of when the fibration condition \eqref{ZGfibrationq}
relating $\{ \cZ_G \}$ to the $\{ \cZ_{G,q} \}$ holds, notice that, although the
recursion is the same for all fixed $q$ (all being the same as for the
case of variable $q$), one does not a priori know whether the seeds of the
recursions are independent of $q$.

\subsection{Edge splitting and the tangent cone}

Again we consider first the operation of splitting one edge in two and
then the case of multiple splittings.
In the case of the tangent cone, we need to distinguish regular
edges from bridges and looping edges.

\begin{lem}\label{splitV2Gbridgeloop}
If the edge $e$ of $G$ is either a bridge or a looping edge and $\bc mG$ denotes the graph
obtained by iterated splitting of the edge $e$, then for $m\geq 1$ the class of the
tangent cone complement satisfies $\{ \cV_{\bc {m+1}G} \} = \Tbb^m \{\cV _G\}$.
\end{lem}

\proof
If $e$ is a bridge, then splitting $e$ amounts to inserting a new 
bridge; by the deletion-contraction formulas for 
bridges (Proposition~\ref{deleconPottstcbloop})
\[
\{ \cV_{\bc 2G} \}= \Tbb\cdot \{\cV _G\}.
\]
Thus, in the case of a bridge one also already sees that
splitting the edge multiple times just has the effect of multiplying 
$\{\cV_G\}$ by a power of $\Tbb$.

If $e$ is a looping edge,
adding a loop to a graph $G'$ multiplies the class $\{\cV_{G'}\}$ by $\Tbb+1$
as seen in Lemma \ref{propVGlist}. 
Attaching a split loop amounts to 
multiplying $Z_{G'}$ (and hence $P_{G'}$) by $(q+t_e+t_f)$, where $t_e$ and
$t_f$ are the variables corresponding to the two new edges; this is simply
because $q^2+t_eq+t_fq$ is the $Z$-polynomial for a $2$-banana. 
Here $G'$ is the graph obtained from $G$ by removing the loop. From
\[
P_{\bc 2G}=P_{G'}\, (q+t_e+t_f)
\]
we see that $P_{\bc 2G}\ne 0$ implies $P_{G'}\ne 0$ and $t_e\ne -(q+t_f)$.

Therefore, for any point $(q,t)$ for which $P_{G'}\ne 0$ we have an
$\Abb^1$ worth of choices for $t_f$ and an $\Abb^1\smallsetminus \Abb^0$
worth of choices for $t_e$. So we can conclude that
\[
\{ \cW_{\bc 2G}\}=\Lbb (\Lbb-1)\cdot \{\cW_{G'}\}=\Tbb \cdot \{ \cW_G \} .
\]
The same analysis goes through after setting $q$ to $0$, and hence
\[
\{ \cY_{\bc 2G}\}=\Lbb (\Lbb-1)\cdot \{\cY_{G'}\}=
\Tbb \cdot \{ \cY_G\}.
\]
The conclusion is that
\[
\{\cV_{\bc 2 G} \}=\Tbb \cdot \{\cV_G\} .
\]
Thus, splitting a loop once has again the effect of
multiplying the corresponding $\{\cV_G\}$ by $\Tbb$, as for bridges.
\endproof

We then check the more interesting case of regular edges.

\begin{thm}\label{splitPottsCone}
Let $e$ be a regular edge of $G$, and let $\bc 2G$ be the graph obtained
by introducing a $2$-valent vertex in $e$, thereby splitting it. With other
notation as above,
\[
\{\cV_{\bc 2G}\}=(\Tbb-2) \{\cV_G\} +(\Tbb-1) \{\cV_{G/e}\}+
(\Tbb+1) (\{\cV_{G\smallsetminus e}\}+ \{ V(Q_{G\smallsetminus e}-q\, Q_{G/e})
\}).
\]
\end{thm}

\proof
Let $e$ be a regular edge, and call $e$ (again) and $f$ the 
two edges created in the process.
The polynomial for $\cW_{\bc 2G}$ is found easily by applying deletion-contraction:
\[
Q_{\bc 2G}=Q_{G\setminus e}\cdot (q+t_e)+t_f \cdot Q_G.
\]
Indeed, deleting $f$ leaves the graph $G\setminus e$ with a dangling
edge $e$ attached, and attaching an edge $e$ to a vertex has the effect
of multiplying the corresponding polynomial by $(q+t_e)$; contracting $f$ gives $G$ back.
By Proposition~\ref{deleconPottstc}, therefore we have
\begin{equation}\label{W2G}
\{ \cW_{\bc 2G}\}=\Lbb \cdot [\Abb^{\#E(G)+1}-\cW_\cap]
-\{ \cW_G \} ,
\end{equation}
where $\cW_\cap$ is the intersection in $\Abb^{\# E(G)+1}$ of the hypersurface
$\cW_G$, with equation $Q_G=0$, and the hypersurface with equation 
$Q_{G\smallsetminus e}(q+t_e)=0$. 

Set theoretically, the locus $\cW_\cap$ is given by
\begin{align*}
\cW_\cap &= V(Q_{G\smallsetminus e}(q+t_e), Q_G)
= V(Q_{G\smallsetminus e}(q+t_e), Q_{G\smallsetminus e}+ t_e Q_{G/e}) \\
&=V(Q_{G\smallsetminus e}, t_e Q_{G/e}) \cup
V(q+t_e, Q_{G\smallsetminus e}+ t_e Q_{G/e}) \\
&=V(Q_{G\smallsetminus e}, Q_{G/e}) \cup V(t_e, Q_{G\smallsetminus e}) \cup
V(q+t_e, Q_{G\smallsetminus e}-q Q_{G/e}) ,
\end{align*}
where all ideals are viewed in $\Abb^{\# E(G)+1}$, with coordinates $q$ and
$t_a$, $a\in E(G)$.
Inclusion--exclusion then gives $[\cW_\cap]$ as the sum of the three loci indicated
here, minus the sum of the three pairwise intersections, plus the triple 
intersection. 
The three pairwise intersections are
\begin{gather*}
V(t_e, Q_{G\smallsetminus e}, Q_{G/e}) \\
V(q+t_e, Q_{G\smallsetminus e}, Q_{G/e}) \\
V(q, t_e, Q_{G\smallsetminus e})
\end{gather*}
and the triple intersection is
\[
V(q,t_e,Q_{G\smallsetminus e}, Q_{G/e})\quad.
\]
This is very similar to the situation we have seen for the case of $\cZ_G$.
We find
\[
[V(Q_{G\smallsetminus e}, Q_{G/e})]
=\Lbb\cdot [\cW_{G\smallsetminus e}\cap \cW_{G/e}]
\]
as $G\smallsetminus e$ and $G/e$ have $E(G)\smallsetminus \{e\}$ as index
set, while
\[
[V(t_e,Q_{G\smallsetminus e}, Q_{G/e})]
=[V(q+t_e,Q_{G\smallsetminus e}, Q_{G/e})]
=[\cW_{G\smallsetminus e}\cap \cW_{G/e}] .
\]
In both cases, the first equation eliminates $t_e$, doing nothing to the rest
since the rest does not depend on $t_e$. We also have
\[
[V(t_e,Q_{G\smallsetminus e})]=[\cW_{G\smallsetminus e}]\quad,\quad
[V(q,t_e,Q_{G\smallsetminus e})]=[\cY_{G\smallsetminus e}],
\]
and
\[
[V(q,t_e,Q_{G\smallsetminus e},Q_{G/e})]=[\cY_{G\smallsetminus e}\cap 
\cY_{G/e}] .
\]
That leaves us with
\[
V(q+t_e, Q_{G\smallsetminus e}-q Q_{G/e})\subseteq \Abb^{\# E(G)+1},
\]
or simply
\[
V(Q_{G\smallsetminus e}-q Q_{G/e})\subseteq \Abb^{\# E(G)},
\]
since the only effect of the first equation is to eliminate $t_e$. 

Applying then inclusion--exclusion we obtain
\[
[\cW_\cap] = (\Lbb-2) [\cW_{G\smallsetminus e}\cap \cW_{G/e}]
+[\cW_{G\smallsetminus e}] + [ V(Q_{G\smallsetminus e}-q\, Q_{G/e})] 
- [\cY_{G\smallsetminus e}]
+[\cY_{G\smallsetminus e}\cap \cY_{G/e}].
\]
Passing to the classes of the complements and implementing the
resulting expression for $\bc 2G$ then gives
\begin{align*}
\{ \cW_{\bc 2G} \}&=
(\Lbb-2)\Lbb(\Lbb^{\#E(G)}-[\cW_{G\smallsetminus e}\cap \cW_{G/e}])
+ \Lbb \{ \cW_{G\smallsetminus e} \} \\
&\qquad
+\Lbb \{ V(Q_{G\smallsetminus e}-q\, Q_{G/e}) \}
-\Lbb \{ \cY_{G\smallsetminus e}]\} 
+\Lbb(\Lbb^{\#E(G)-1}-[\cY_{G\smallsetminus e}\cap 
\cY_{G/e}]) -\{ \cW_G \}.
\end{align*}
Using Proposition~\ref{deleconPottstc} and computing
the difference $\{ \cV_G \} = \{ \cW_G \} - \{ \cY_G \}$ one gets
\begin{align*}
\{\cV_{\bc 2G}\} &=\{ \cW_{\bc 2G} \}-\{ \cY_{\bc 2G}\} \\
&=(\Lbb-3) \{ \cW_G \}+ (\Lbb-2)\{ \cW_{G/e} \}
+ \Lbb (\{\cW_{G\smallsetminus e}\}- \{\cY_{G\smallsetminus e}\})
+\{\cY_G\} \\
&\qquad
+\{\cY_{G/e}\}
+\Lbb \{ V(Q_{G\smallsetminus e}-q\, Q_{G/e}) \}
-(\Lbb-2) \{ \cY_G \}-(\Lbb-1) \{ \cY_{G/e}\} 
-\Lbb \{\cY_{G\smallsetminus e}\} \\
&=(\Lbb-3)(\{\cW_G\}-\{\cY_G\})
+(\Lbb-2) (\{\cW_{G\smallsetminus e}\}-\{\cY_{G\smallsetminus e}\})\\
&\qquad
+\Lbb(\{ \cW_{G\smallsetminus e}\}- \{\cY_{G\smallsetminus e}\})
+\Lbb(\{ V(Q_{G\smallsetminus e}-q\, Q_{G/e}) \}
- \{ \cY_{G\smallsetminus e} \}) \\
&=(\Lbb-3) \{\cV_G\} +(\Lbb-2) \{\cV_{G/e}\}+\Lbb \{\cV_{G\smallsetminus e}\}
+\Lbb \{ V(Q_{G\smallsetminus e}-q\, Q_{G/e}) \} .
\end{align*}
\endproof

As in the case of $\cZ_G$ and $\cZ_{G,q}$, the polynomial
$Q_{G\smallsetminus e}-q\, Q_{G/e}$ can be given a combinatorial
interpretation as
$$ Q_{G\smallsetminus e}-q\, Q_{G/e} =
\sum_{\text{$A\subseteq E\smallsetminus \{e\}$ 
forest connecting endpoints of $e$}} q^{k(A)-k(C)} \prod_{a\in A} t_a. $$

We then can proceed as in the case of $\cZ_G$ and $\cZ_{G,q}$ and
obtain the following recursions.

\begin{prop}\label{newdc}
With notation as above, and assuming $e$ is not a looping edge,
\[
\{ \cY_{\bc {m+3}G}\} =
(2\Tbb-1) \{ \cY_{\bc {m+2}G} \} - \Tbb(\Tbb-2) \{ \cY_{\bc {m+1}G}\} 
-\Tbb^2 \{\cY_{\bc mG}\}
\]
for all $m\ge 0$. The generating function 
\[
G^{\cY}_e(s):= \sum_{m\ge 0} \{ \cY_{\bc mG}\} \frac{s^m}{m!}
\]
for the classes $\{ \cY_{\bc mG} \}$
is then of the form
\[ 
G^{\cY}_e(s)=A\, e^{-s}+B\, s\,e^{\Tbb s}+C\, e^{\Tbb s},
\]
with the constants $A$, $B$, and $C$ satisfying
\begin{align*}
A &=\{ \cY_{\bc 0G}\}
-2\frac{\{\cY_{\bc 1G}\}+\{\cY_{\bc 0G}\}}{\Tbb+1}
+\frac{\{\cY_{\bc 2G}\}+2\{\cY_{\bc 1G}\}
+\{\cY_{\bc 0G}\}}{(\Tbb+1)^2} \\
B &=-\{\cY_{\bc 1G}\}-\{\cY_{\bc 0G}\}
+\frac{\{\cY_{\bc 2G}\}+2\{\cY_{\bc 1G}\}
+\{\cY_{\bc 0G}\}}{\Tbb+1} \\
C &=2\frac{\{\cY_{\bc 1G}\} +\{\cY_{\bc 0G}\}}{\Tbb+1}
-\frac{\{\cY_{\bc 2G}\}+2\{\cY_{\bc 1G}\}
+\{\cY_{\bc 0G}\}}{(\Tbb+1)^2} 
\end{align*}
Similarly,
\[
\{\cV_{\bc {m+3}G}\}=(2\Tbb-2) \{\cV_{\bc {m+2}G}\} -(\Tbb^2-3\Tbb+1) \{\cV_{\bc {m+1}G}\}
-\Tbb(\Tbb-1) \{\cV_{\bc mG}\},
\]
for all $m\ge 0$, with generating function
\[ 
G^{\cV}_e(s):= \sum_{m\ge 0} \{\cV_{\bc mG}\}\frac{s^m}{m!}
\]
given by
\[
G^{\cV}_e(s)=A\, e^{-s}+B\, e^{\Tbb s}+C\,e^{(\Tbb-1) s}
\]
with the terms $A$, $B$ and $C$ satisfying
\begin{align*}
A &=\{\cV_{\bc 0G}\}
+\frac{\{\cV_{\bc 2G}\}+\{\cV_{\bc 1G}\}}{\Tbb}
-\frac{\{\cV_{\bc 2G}\}+3\{\cV_{\bc 1G}\}+2\{\cV_{\bc 0G}\}}{\Tbb+1}\\
B &=-\{\cV_{\bc 1G}\}-\{\cV_{\bc 0G}\}
+\frac{\{\cV_{\bc 2G}\}+3\{\cV_{\bc 1G}\}+2\{\cV_{\bc 0G}\}}{\Tbb+1}\\
C &=\{\cV_{\bc 1G}\}+\{\cV_{\bc 0G}\}
-\frac{\{\cV_{\bc 2G}\}+\{\cV_{\bc 1G}\}}{\Tbb}.
\end{align*}
\end{prop}

The argument is essentially analogous to the cases of $\cZ_G$ and $\cZ_{G,q}$
analyzed before and we do not reproduce it explicitly here.

The generating function for $\cY_{\bc mG}$ is `dual' to the one for $\cX_{\bc mG}$
given in \cite{AluMa3}: the effect of splitting edges on the class $\{\cY_{\bc mG}\}$
is analogous to the effect of multiplying edges on the class $\{\cX_{\bc mG}\}$.

\section{Polygons and linked polygons}

We now focus on a particularly simple class of graphs for which
we can compute everything explicitly. These will be polygons and
graphs constructed out of chains of linked polygons. We will later
focus especially on this class of graphs to provide an explicit
example of how to apply these motivic techniques to analyze
(virtual) phase transitions in the corresponding Potts models.

We start by using the formulae for edge splittings obtained in
the previous section to compute explicitly
the classes $\{ \cZ_G \}$ for polygon graphs.

\begin{prop}\label{polygonZGclass}
Let $\bc mG$ be an $(m+1)$-sided polygon. Then the classes $\{ \cZ_{\bc mG} \}$ are given 
explicitly by the formula
\begin{equation}\label{polygonclass}
\{\cZ_{\bc mG}\} = \Tbb^{m+2} + \Tbb(\Tbb-1)(\Tbb^m-(\Tbb-1)^m)
+(\Tbb-1)\frac{(\Tbb-1)^m-(-1)^m}\Tbb\quad.
\end{equation}
\end{prop}

\begin{proof} Let us check directly the initial cases that are needed to use the
recursive formula for edge splittings. The first graphs are
$\bc 0G$ a single loop, $\bc 1G$ a $2$-banana (two vertices with two parallel
edges between them), $\bc 2G$ a triangle. The equations $Z_G(q,t)=0$ 
are of the form
\begin{align*}
\bc 0G \quad &:\quad q+qt=0\quad \text{(in $\Abb^{1+1}$)} \\
\bc 1G \quad &:\quad q^2 + (t_1+t_2+t_1t_2) q = 0 \quad \text{(in $\Abb^{2+1}$)} \\
\bc 2G \quad &:\quad q^3 + (t_1+t_2+t_3)q^2 + (t_1 t_2 + t_1 t_3 +t_2 t_3 +t_1 t_2 t_3)q
=0 \quad \text{(in $\Abb^{3+1}$)}
\end{align*}
The corresponding classes $\{\cZ_G\}$ can then be computed directly in these
cases by
applying the basic facts listed in Corollary \ref{listZGprop} and Theorem \ref{resplitPotts}. 
One obtains
\begin{align*}
\{Z_{\bc 0G}\} &= \Tbb^2 \\
\{Z_{\bc 1G}\} &= \Tbb^3+\Tbb^2-1 \\
\{Z_{\bc 2G}\} &= \Tbb^4+2\Tbb^3-2\Tbb^2-2\Tbb+2 .
\end{align*}
The expression \eqref{polygonclass} is then $m!$ times the coefficient 
of $s^m$ in the expansion of the right-hand
side of the generating function \eqref{ZmGgf}, with these initial conditions.
\end{proof}

\subsection{Polygons at fixed $q$}

We will also need in the following the classes of the polygon graphs for
fixed value of $q$. These are obtained as follows.

\begin{prop}\label{polygonZGclassq}
Let $\bc mG$ be an $(m+1)$-sided polygon. Then the classes $\{ \cZ_{\bc mG,q} \}$ 
for fixed $q\neq 0,1$ are given  explicitly by the formula
\begin{equation}\label{polygonZGq}
\{\cZ_{\bc mG,q}\}
=\Tbb^{m+1}+ \Tbb (\Tbb^m-(\Tbb-1)^m)+\frac{(\Tbb-1)^m-(-1)^m}\Tbb .
\end{equation}
The polygon graphs $\bc mG$ satisfy the fibration condition  \eqref{ZGfibrationq}.
\end{prop}

\proof
As above, the seeds of the recursion are a single loop,
a $2$-banana, and a triangle.

The single loop has class $\Tbb$ from Corollary \ref{listZGqprop}.

The $2$-banana hypersurface has equation $q t_1 t_2 + q t_1 + q t_2 + q^2=0$.
Factoring out a $q$ (assumed to be nonzero to begin with) this is equivalent
to
\[
(t_1+1)(t_2+1)=1-q .
\]
We are assuming that $q\ne 1$, so up to a variable change this equation is
\[
u_1 u_2 = r \ne 0 .
\]
This forces $u_1\ne 0$, and determines $u_2$ once $u_1$ is fixed; thus the
class of this locus in the Grothendieck group is $\Lbb-1$; its complement in 
$\Abb^2$ has class
\[
\Lbb^2 - \Lbb + 1 = \Tbb^2 + \Tbb + 1 .
\]
This is independent of $q\ne 0,1$.

The triangle hypersurface has equation
\[
q(t_1 t_2 t_3 + t_1 t_2 + t_1 t_3 + t_2 t_3 + q (t_1 + t_2 +t_3) + q^2) = 0 .
\]
Changing variables: $u_i=t_i+1$, $r=q-1$, and keeping in mind $q\ne 0$, this
is equivalent to
\begin{equation}\label{u1u2u3}
u_1 u_2 u_3 + (u_1+u_2+u_3)r =r(1-r) .
\end{equation}
Solving for $u_3$ gives
\[
u_3=\frac{r(1-r-u_1-u_2)}{u_1 u_2+r} .
\]
Thus, the variety contains an open subvariety isomorphic to
\[
\Abb^2 \smallsetminus V(u_1 u_2+r)\quad;
\]
this locus has class $\Lbb^2-\Lbb +1$, as we just computed above, since
$r\ne 0$ (as $q\ne 1$).
If $u_1 u_2=-r$, then \eqref{u1u2u3} is equivalent to
\[
u_1+u_2=1-r ,
\]
implying easily that $(u_1,u_2)$ equals $(1,r)$ or $(r,1)$, while $u_3$ is free
in this case. That is, the complement in \eqref{u1u2u3} of the open subvariety 
determined above consists of the loci
\[
\{(1,-r,u_3)\}\cong \Abb^1\quad, \quad \{(-r,1,u_3)\}\cong \Abb^1 .
\]
These lines are distinct, since $r\ne -1$ (as $q\ne 0$). The conclusion is that
the class of \eqref{u1u2u3} equals
\[
\Lbb^2-\Lbb+1 + 2\Lbb = \Lbb^2 + \Lbb + 1\quad,
\]
and hence its complement in $\Abb^3$ has class
\[
\Lbb^3-\Lbb^2-\Lbb-1 = \Tbb^3+2\Tbb^2-2\quad.
\]
Again this is independent of $q$.

These cases are compatible with the fibration condition 
\eqref{ZGfibrationq}. For instance, in the triangle case, we get
\[
(\Tbb-1)(\Tbb^3+2\Tbb^2-2)+\Tbb^3 = \Tbb^4+2\Tbb^3-2\Tbb^2-2\Tbb+2
\]
in agreement with the class for a triangle in the free $q$ case, used in
Proposition \ref{polygonZGclass}. 

The recursion then gives the classes \eqref{polygonZGq}, 
\begin{align*}
\{\cZ_{\bc mG,q}\}
&=\frac{\Tbb^{m+2} + \Tbb(\Tbb-1)(\Tbb^m-(\Tbb-1)^m)
+(\Tbb-1)\frac{(\Tbb-1)^m-(-1)^m}\Tbb - \Tbb^{m+1}}
{\Tbb-1} \\
&=\Tbb^{m+1}+ \Tbb (\Tbb^m-(\Tbb-1)^m)+\frac{(\Tbb-1)^m-(-1)^m}\Tbb\quad.
\end{align*}

One can
verify explicitly the compatibility of  \eqref{polygonZGq} and \eqref{polygonclass}
with the fibration condition  \eqref{ZGfibrationq}: the polynomial 
\[
\{Z_{\bc mG}\} = \Tbb^{m+2} + \Tbb(\Tbb-1)(\Tbb^m-(\Tbb-1)^m)
+(\Tbb-1)\frac{(\Tbb-1)^m-(-1)^m}\Tbb
\]
factors exactly as predicted by  \eqref{ZGfibrationq} in terms of the classes
\eqref{polygonZGq}.
\endproof

\subsection{Polygons and the class of the tangent cone}

The explicit formula for the classes $\{ \cV_G \}$ of the complement
of the tangent cone of the Potts model hypersurface for polygons is
obtained as follows.

\begin{prop}\label{polygcone}
Let $\bc mG$ be the polygon with $m+1$ edges. Then
\[
\{\cV_{\bc mG}\}=(\Tbb-1)(-1)^m + 2\Tbb^{m+2}-(\Tbb+1)(\Tbb-1)^{m+1}.
\]
\end{prop}

\proof
For the seed of the recursion we have in this case
\begin{align*}
\bc 0G \quad &:\quad q=0\quad \text{(in $\Abb^{1+1}$)} \\
\bc 1G \quad &:\quad q^2 + (t_1+t_2) q = 0 \quad \text{(in $\Abb^{2+1}$)} \\
\bc 2G \quad &:\quad q^3 + (t_1+t_2+t_3)q^2 + (t_1 t_2 + t_1 t_3 +t_2 t_3)q
=0 \quad \text{(in $\Abb^{3+1}$)}
\end{align*}
for which we then get
\begin{align*}
\{\cV_{\bc 0G}\} &= \Lbb^2-\Lbb^1 = \Tbb(\Tbb+1) \\
\{\cV_{\bc 1G}\} &= \Lbb^3-2\Lbb^2+\Lbb^1 = \Tbb^2(\Tbb+1) \\
\{\cV_{\bc 2G}\} &= \Tbb^4+2\Tbb^3-\Tbb = \Tbb(\Tbb+1)(\Tbb^2+\Tbb-1).
\end{align*}
The first two expressions are immediate; for the third, note that $\cV_{\bc 2G}$
is a cone over a nonsingular quadric in $\Pbb^3$; the computation is then
straightforward. We then have
\begin{multline*}
G^{\cV}_e(s)=e^{-s}\left(
\left(1+\frac{e^{\Tbb s}-1}{\Tbb}\right) \Tbb(\Tbb+1)
+\frac{e^{\Tbb s}-1}{\Tbb} \Tbb^2(\Tbb+1) \right. \\
+ \left(e^{(\Tbb+1)s}-e^{\Tbb s}-\frac{e^{\Tbb s}-1}\Tbb\right)\cdot 2 \Tbb^2
\bigg)\quad.
\end{multline*}
that is
\[
G^{\cV}_e(s)=(\Tbb-1) e^{-s}
+2\Tbb^2 e^{\Tbb s}
-(\Tbb^2-1) e^{(\Tbb-1)s}\quad.
\]
Reading off the coefficient of $s^m/m!$ then gives the result.
\endproof

\subsection{Chains of linked polygons}\label{chainmkNsec}

A class of graphs that have been extensively studied from the point of view of
the statistical mechanics of Potts models is the case of chains of linked polygons.
These are graphs consisting of $N$ equal polygons $\bc mG$, each with $m+1$ edges,
attached to one another by chains of $k\geq 0$ edges. The case $k=0$ corresponds to
polygons joined at vertices. The corresponding Potts models were studied, from the
point of view of the properties of ground state entropy, in \cite{ShTs}: see also the
references therein for several other results on this class of Potts models.
\begin{center}
\includegraphics[scale=.5]{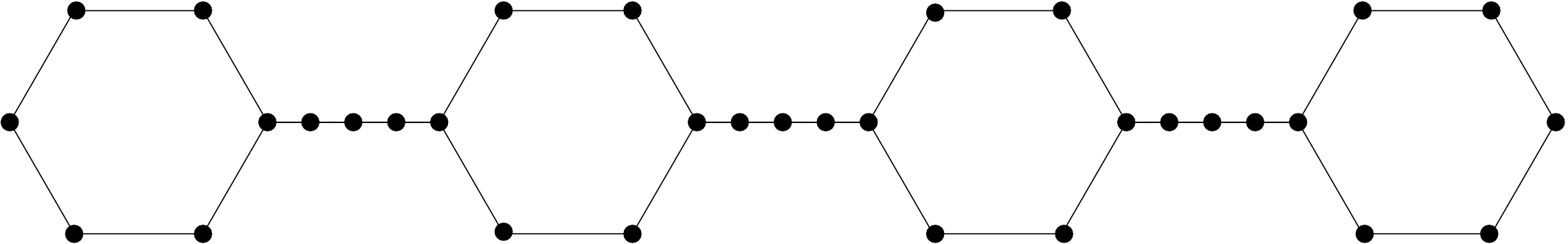}
\end{center}

\begin{defn}\label{GmkN}
Let $\bc {(m,k)}G^N$ denote the graph obtain by joining $N$
polygons, each with $m+1$ sides, with every two nearby polygons
connected by a chain of $k\geq 0$ edges.
\end{defn}

This family of graphs has three parameters $m$, $k$, and $N$, each of
which can be independently sent to $\infty$ to create an infinite graph.
Thus, we will use this family as our main example on which to analyze
the topological complexity of the corresponding set of virtual phase
transitions.

The classes $\{ \cZ_{G,q} \}$ can be computed explicitly for this family from
the result for polygons in Proposition \ref{polygonZGclassq} and the other 
basic properties of Corollary \ref{listZGqprop}.

\begin{prop}\label{classGmkNq}
For $\bc {(m,k)}G^N$ as in Definition \ref{GmkN}, the classes 
$\{\cZ_{G,q}\}$ with fixed $q\ne 0,1$
are given by
\begin{equation}\label{ZGmkNq}
\left( \Tbb^{m+1}+ \Tbb (\Tbb^m-(\Tbb-1)^m)+\frac{(\Tbb-1)^m-(-1)^m}\Tbb \right)^N \bT^{k(N-1)}.
\end{equation}
\end{prop}

\proof By (3) and (4) of Corollary \ref{listZGqprop}, we have
$$ \{ \cZ_{\bc {(m,k)}G^N,q} \} = \{ \cZ_{\bc mG,q} \}^N \,\, \bT^{k(N-1)}. $$
The result then follows from \eqref{polygonZGq}.
\endproof

In this case, since the graphs involved satisfy the fibration condition \eqref{ZGfibrationq},
one can also obtain an explicit formula for the classes $\{ \cZ_{\bc {(m,k)}G^N} \}$
with variable $q$, from \eqref{disjunfibr} of Corollary \ref{disjunfibrcor}.

\section{Multiple edge formula} 

The operation of doubling edges is dual to splitting edges. However, while in the
case of the graph hypersurfaces of Feynman graphs analyzed in \cite{AluMa3}
the corresponding operation on the hypersurface complement classes in the
Grothendieck ring is very simple, this is not the case when one considers the
Potts model hypersurface.

In fact, the combinatorial deletion--contraction formula for the multivariate
Tutte polynomial shows that the polynomial for the graph obtained from $G$
by doubling an edge $e$ is of the form
\begin{equation}\label{2edgeZGqt}
Z_{G\smallsetminus e} + (t_e+t_f+t_e t_f) Z_{G/e} ,
\end{equation}
and it is the presence of the extra term $t_e t_f$ here that complicates the matter.

\subsection{Edge doubling}
We derive here a formula for the class $\{ \cZ_G \}$ under the operation of
doubling an edge, and then we obtain a recursive formula for the iteration 
of this operation.

\begin{thm}\label{doubPotts}
Let $G'$ be the graph obtained by doubling the edge $e$ in a graph $G$.
Then
\begin{equation}\label{2edgeZG}
\{\cZ_{G'}\} = \Tbb\cdot \{\cZ_G\} + (\Tbb +1 )\cdot \{B^e_G\} ,
\end{equation}
where $B^e_G$ is the locus of zeros of $Z_{G\smallsetminus e}-Z_{G/e}$.
\end{thm}

\proof
It is convenient to change variables, letting $u_e=1+t_e$, $u_f=1+t_f$.
Then the class \eqref{2edgeZGqt} for the double edged graph is given by
\begin{equation}\label{2edgeZGqt2}
Z_{G\smallsetminus e} + (u_e u_f-1) Z_{G/e}\ne 0 .
\end{equation}

If $Z_{G/e}=0$, then necessarily $Z_{G\smallsetminus e}\ne 0$; $u_e$
and $u_f$ are free, so this accounts for a class
\[
(\Tbb+1)^2\cdot [\cZ_{G/e}\smallsetminus (\cZ_{G\smallsetminus e}\cap \cZ_{G/e})] .
\]

If $Z_{G/e}\ne 0$, then the condition amounts to
\[
u_e u_f \ne 1 -\frac{Z_{G\smallsetminus e}}{Z_{G/e}} .
\]
This in turn leads to to two possibilities:
\begin{itemize}
\item Either $\frac{Z_{G\smallsetminus e}}{Z_{G/e}} = 1$, and then 
$u_e u_f\ne 0$; this accounts for $\Lbb^2-2\Lbb +1=\Tbb^2$;
\item Or $\frac{Z_{G\smallsetminus e}}{Z_{G/e}} \ne 1$, and then 
$u_e u_f\ne c$ for some $c\ne 0$. For $c\ne 0$, $u_e u_f=c$ necessarily
gives $u_f\ne 0$, $u_e=c/u_f$; this accounts for $\Lbb-1$. Thus the class
of $u_e u_f\ne c$ is $\Lbb^2-\Lbb+1=\Tbb^2+\Tbb+1$.
\end{itemize}

In total, the class of the complement for the double-edged graph is
\begin{multline*}
(\Tbb+1)^2\cdot [\cZ_{G/e}\smallsetminus (\cZ_{G\smallsetminus e}\cap \cZ_{G/e})]
+\Tbb^2 [(\Abb^{|E|}\smallsetminus \cZ_{G/e}) \cap 
V(Z_{G\smallsetminus e}-Z_{G/e})] \\
+(\Tbb^2+\Tbb+1) [(\Abb^{|E|}\smallsetminus \cZ_{G/e}) \smallsetminus 
V(Z_{G\smallsetminus e}-Z_{G/e})] .
\end{multline*}

This expression can be further simplified in the following way:
\begin{multline*}
(\Tbb+1)^2\cdot [\cZ_{G/e}\smallsetminus (\cZ_{G\smallsetminus e}\cap \cZ_{G/e})]
+(\Tbb^2+\Tbb+1) [\Abb^{\#E(G)}\smallsetminus \cZ_{G/e}] \\
-(\Tbb+1) [(\Abb^{\#E(G)}\smallsetminus \cZ_{G/e}) \cap 
V(Z_{G\smallsetminus e}-Z_{G/e})] 
\end{multline*}
\begin{multline*}
=(\Tbb^2+2\Tbb+1)\cdot [\cZ_{G/e}]-(\Tbb+1)^2 [\cZ_{G\smallsetminus e}\cap \cZ_{G/e}]
+(\Tbb^2+\Tbb+1) [\Abb^{\#E(G)}\smallsetminus \cZ_{G/e}] \\
-(\Tbb+1) [(\Abb^{\#E(G)}\smallsetminus \cZ_{G/e}) \cap 
V(Z_{G\smallsetminus e}-Z_{G/e})] 
\end{multline*}
\begin{multline*}
=\Tbb \cdot [\cZ_{G/e}]-(\Tbb+1)^2 [\cZ_{G\smallsetminus e}\cap \cZ_{G/e}]
+(\Tbb^2+\Tbb+1) \Lbb^{\#E(G)} \\
-(\Tbb+1) [(\Abb^{\#E(G)}\smallsetminus \cZ_{G/e}) \cap 
V(Z_{G\smallsetminus e}-Z_{G/e})] 
\end{multline*}
\begin{multline*}
=(\Tbb+1)^2 ([\Lbb^{\#E(G)}-[\cZ_{G\smallsetminus e}\cap \cZ_{G/e}])
-\Tbb ([\Lbb^{\#E(G)}-[\cZ_{G/e}]) \\
-(\Tbb+1) [(\Abb^{\#E(G)}\smallsetminus \cZ_{G/e}) \cap 
V(Z_{G\smallsetminus e}-Z_{G/e})] .
\end{multline*}
We then use Theorem~\ref{delconZGthm} to express $\{\cZ_{G\smallsetminus e}
\cap \cZ_{G/e}\}$, obtaining
\[
(\Tbb+1) (\{\cZ_G\}+\{\cZ_{G/e}\})
-\Tbb (\{\cZ_{G/e}\})
-(\Tbb+1) [(\Abb^{\#E(G)}\smallsetminus \cZ_{G/e}) \cap 
V(Z_{G\smallsetminus e}-Z_{G/e})] 
\]
and hence simply 
\[
(\Tbb+1) \{\cZ_G\}+\{\cZ_{G/e}\} - (\Tbb+1)  [(\Abb^{\# E(G)}\smallsetminus \cZ_{G/e}) 
\cap V(Z_{G\smallsetminus e}-Z_{G/e})] .
\]
Now
\[
(\Abb^{\#E(G)}\smallsetminus \cZ_{G/e}) \cap V(Z_{G\smallsetminus e}-Z_{G/e})
=V(Z_{G\smallsetminus e}-Z_{G/e}) \smallsetminus
(\cZ_{G\smallsetminus e}\cap \cZ_{G/e}),
\]
and again using Theorem~\ref{delconZGthm}  we get that the class of the
complement for the double-edged graph is
\[
\Tbb\cdot \{\cZ_G\}- (\Tbb+1)\cdot \{ V(Z_{G\smallsetminus e}-Z_{G/e}) \} .
\]
\endproof

As in the case of the term $A^e_G$ in the formula for edge splitting, the term $B^e_G$
here also has an interpretation in terms of the combinatorics of the graph.

\begin{lem}\label{WeGlocus}
The locus $B^e_G$ of zeros of $Z_{G\smallsetminus e}-Z_{G/e}$ is
equivalently the locus of zeros of a polynomial $(q-1) \overline Z''$,
where $Z''$ is the sum of the monomials $\prod_{e\in A} t_e$ corresponding to 
subgraphs $A$ of $G/e$ that acquire an additional connected component
when they are viewed in $G\smallsetminus e$,
and $\overline Z''=Z''/q$.
\end{lem}

\proof
Recall from \eqref{bijectioneq} that
\[
Z_{G\smallsetminus e}=Z' + Z''\quad, \quad Z_{G/e}=Z' + \frac {Z''}q\quad.
\]
and hence
\[
Z_{G\smallsetminus e}-Z_{G/e} = (q-1) \overline Z''\quad,
\]
where $\overline Z''=Z''/q$ is indeed the sum of the standard monomials over the
subgraphs of $G/e$ which acquire an additional connected component
when they are viewed in $G\smallsetminus e$.
\endproof

Thus, the description of $B^e_G$ is somewhat complementary
to that of $A^e_G$: both are $(q-1)$ times a sum of terms having to do
with subgraphs of $G/e$. For $A$, one looks at graphs for which the
number of connected components is the same when the subgraph is
viewed in $G\smallsetminus e$; for $B$, one looks at the graphs for
which the number of connected components increases.

\begin{ex}{\rm
When $e$ is a looping edge, then $B^e_G=0$: indeed, $G\smallsetminus e
=G/e$ in this case. Thus $\{B^e_G\}=0$, and the formulas simplifies to
$\{Z_{G'}\}= \Tbb\cdot \{Z_G\}$, which recovers the case (5) of Corollary \ref{listZGprop},
namely attaching a new looping edge to $G$. }
\end{ex}

\subsection{Multiple edge formulas}
We now consider the case where the operation of doubling an edge
is iterated, that is, where an edge $e$ in a graph is replaced by $m$
parallel edges, between the same endpoints. This can be seen also
as replacing an edge $e$ with the $m$-th banana graph (see \cite{AluMa}).

\begin{thm}\label{reremuled}
Let $G^{(m)}$ denote the graph obtained by adding $m$ edges parallel
to $e$ in $G$. (So $G=G^{(0)}$, $G'=G^{(1)}$.) Then, for $m\ge 0$, the
classes $\{ \cZ_{G^{(m)}} \}$ satisfy
\begin{equation}\label{medgesZG}
\{\cZ_{G^{(m+2)}}\} = (2\Tbb+1) \{\cZ_{G^{(m+1)}}\} - \Tbb(\Tbb+1) \{\cZ_{G^{(m)}}\} .
\end{equation}
\end{thm}

\proof
The key case to consider
is that in which we {\em triple\/} a given edge of $G$: let $G'$ denote
(as in Proposition~\ref{doubPotts}) the graph obtained from $G$ by doubling
$e$, and let $G''$ be the graph obtained from $G'$ by doubling $e$ again. Applying
Theorem~\ref{doubPotts} yields
\[
\{\cZ_{G''}\} = \Tbb\cdot \{\cZ_{G'}\} + (\Tbb +1 )\cdot \{B^e_{G'}\} .
\] 
Thus, we have to understand $\{B^e_{G'}\}$. According to Lemma \ref{WeGlocus},
this hypersurface has equation $(q-1) \overline Z''$, where $\overline Z''$
collects monomial according to the subgraphs of $G'/e$ which acquire
a component when viewed in $G'\smallsetminus e$.
The new edge in $G'$
parallel to $e$ cannot be part of any such subgraphs, since it does join
the endpoints of $e$, so it prevents a new component from forming as
we remove $e$. Therefore, the $\overline Z''$ for~$G'$ actually equals on the nose
the $\overline Z''$ for $G$; the only difference between $B^e_G$ and $B^e_{G'}$
is that the latter is contained in a space of dimension one higher, and it
may be described as a cylinder on $B^e_G$. 

Thus, we have
\[
\{\cZ_{G''}\} = \Tbb\cdot \{\cZ_{G'}\} + (\Tbb +1 )\cdot \{B^e_{G'}\}
\]
and $\{W^e_{G'}\}=(\Tbb+1) \{W^e_G\}$.  By Theorem~\ref{doubPotts},
\[
(\Tbb+1) \{W^e_G\} = \{\cZ_{G'}\}-\Tbb \{\cZ_G\} .
\]
Thus, we obtain
\begin{align*}
\{\cZ_{G''}\} &= \Tbb \{\cZ_{G'}\} + (\Tbb +1 ) (\{\cZ_{G'}\}-\Tbb \{\cZ_G\}) \\
&= (2\Tbb+1) \{\cZ_{G'}\} - \Tbb(\Tbb+1) \{\cZ_{G}\} .
\end{align*}
The stated formula follows by applying this formula to $G^{(m)}$ instead of
$G$.
\endproof

We can form a generating function for the classes 
of graphs with multiple edges.

\begin{thm}\label{medgeZGgenfunc}
The generating function of the classes $\{Z_{G^{(m)}}\}$ is given by
\begin{equation}\label{medgesZGgf}
\sum_{m\ge 0} \{\cZ_{G^{(m)}}\}\frac{s^m}{m!} 
= \left((\Tbb+1) \{\cZ_G\}-\{\cZ_{G'}\}\right) e^{\Tbb s} 
+ \left(\{\cZ_{G'}\}-\Tbb \{\cZ_G\}\right) e^{(\Tbb+1) s} .
\end{equation}
\end{thm}

\proof
The recurrence relation \eqref{medgesZG} of
Theorem~\ref{reremuled} translates into the
differential equation
\begin{equation}\label{medgeseqdiff}
g''(s)=(2\Tbb+1) g'(s)-\Tbb(\Tbb+1) g(s)
\end{equation}
for the generating function 
\begin{equation}\label{medgesZGgf2}
g(s)=\sum \{\cZ_{G^{(m)}}\} \frac{s^m}{m!}, 
\end{equation}
solving which reveals that
\[
\sum_{m\ge 0} \{\cZ_{G^{(m)}}\} \frac{s^m}{m!} 
= A e^{\Tbb s} + B e^{(\Tbb+1) s} .
\]
Solving for the constants $A$, $B$ in terms of $\{\cZ_G\}$ and $\{\cZ_{G'}\}$ gives
\eqref{medgesZGgf}.
\endproof

\begin{ex}\label{exbanana}{\rm 
The $m$-th banana graph is a graph with two vertices and $m$ parallel edges
between them.
To compute the class for bananas, we can start with $G$ a single non-looping
edge, for which $\{\cZ_G\}=\Tbb^2$, and $G'$ a $2$-banana, for which
$\{\cZ_{G'}\}=\Tbb^3 + \Tbb^2-1$. Then from \eqref{medgesZGgf} we have
\[
\sum_{m\ge 0} \{\cZ_{G^{(m)}}\} \frac{s^m}{m!} 
= e^{\Tbb s} + (\Tbb^2-1) e^{(\Tbb+1)s}\quad,
\]
from which we obtain
\begin{equation}\label{ZGmbanana}
\{\cZ_{G^{(m)}}\} = \Tbb^m + (\Tbb-1)(\Tbb+1)^{m+1} 
\end{equation}
for the class of the $m$-th banana graph.
A generating function for the analogous class for graph hypersurfaces $\cX_G$ 
is given in \cite{AluMa3}.}
\end{ex}

\subsection{Multiple edges for fixed $q$}

Again the argument for variable $q$ carries over almost identically to cover
the case with fixed $q\neq 0,1$. We obtain the following results.

\begin{prop}\label{doubPottsq}
Let $G'$ be the graph obtained by doubling the edge $e$ in a graph $G$.
Then
\[
\{\cZ_{G',q}\} = \Tbb\cdot \{\cZ_{G,q}\} + (\Tbb +1 )\cdot \{W^e_{G,q}\} ,
\]
where the locus $W^e_{G,q}\subset \Abb^{\#E(G)-1}$, 
for $q\ne 0,1$, is given by the vanishing of the polynomial $Z''$
adding monomials over the subgraphs of $G/e$ which 
acquire an additional connected component when they are viewed in 
$G\smallsetminus e$.
\end{prop}

\begin{cor}\label{reremuledq}
Let $G^{(m)}$ denote the graph obtained by adding $m$ edges parallel
to $e$ in $G$. (So $G=G^{(0)}$, $G'=G^{(1)}$.) Then for $m\ge 0$
\[
\{\cZ_{G^{(m+2)},q}\} = (2\Tbb+1) \{\cZ_{G^{(m+1)},q}\} - \Tbb(\Tbb+1) \{\cZ_{G^{(m)},q}\} .
\]
\end{cor}

The general solution of this recursion also matches the one for free-$q$:
\[
\sum_{m\ge 0} \{\cZ_{G^{(m)},q}\}\frac{s^m}{m!} 
= \left((\Tbb+1) \{\cZ_{G,q}\}-\{\cZ_{G',q}\}\right) e^{\Tbb s} 
+ \left(\{\cZ_{G',q}\}-\Tbb \{Z_{G,q}\}\right) e^{(\Tbb+1) s} .
\]

\begin{ex}\label{bananasq} {\rm Consider the case of the {\em banana graphs}.
The seeds for bananas are a single non-looping edge, with class $\Tbb$
and the $2$-banana, with class $\Tbb^2+\Tbb+1$, as computed above.
Plugging into the last formula, we get the generating function for the classes 
of Potts model complements of banana graphs for fixed $q$:
\[
\sum_{m\ge 0} \{\cZ_{G^{(m)},q}\} \frac{s^m}{m!} = (\Tbb+1)e^{(\Tbb+1)s}-e^{\Tbb s}
\quad;
\]
extracting the term of degree $m$ gives the very simple class for the 
$(m+1)$-banana:
\begin{equation}\label{Zqbananam}
(\Tbb+1)^{m+1}-\Tbb^m
\end{equation}
in agreement with the fibration condition \eqref{ZGfibrationq} for \eqref{ZGmbanana}, 
and in particular independent of $q\ne 0,1$. }
\end{ex}

\subsection{Chains of linked banana graphs}

By analogy to the example of the chains of linked polygon graphs considered in
\S \ref{chainmkNsec}, we consider here a similar family but with the polygons
replaced by banana graphs.
\begin{center}
\includegraphics[scale=.5]{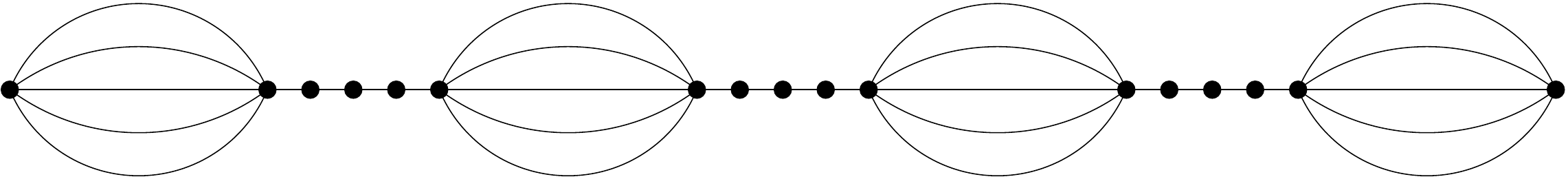}
\end{center}

\begin{defn}\label{chainedbananas}
The graphs $\bc kG^{(m),N}$ are obtained by connecting $N$ banana graphs $G^{(m)}$,
each with $m$ parallel edges, each connected to the next by a chain of $k\geq 0$ edges
(connected by joining vertices in the case $k=0$).
\end{defn}

We can compute the classes $\{ \cZ_{G,q} \}$ for this family of graphs using
the explicit formula \eqref{Zqbananam} for the banana graphs.

\begin{prop}\label{chainbananaclass}
Let $\bc kG^{(m),N}$  be the graphs of Definition \ref{chainedbananas}. Then the
corresponding classes are given by
\begin{equation}\label{ZGqchban}
\{ \cZ_{\bc kG^{(m),N},q} \} =((\bT+1)^{m+1}-\bT^m)^N  \bT^{k (N-1)}.
\end{equation}
\end{prop}

\proof The result follows immediately by applying (3) and (4) of Corollary
\ref{listZGqprop} to the explicit formula \eqref{Zqbananam} in Example \ref{bananasq}.
\endproof

\subsection{Polygon chains}

A class of graphs that can be obtained by alternating edge splitting
and edge doubling operations in different orders are the chains of polygons of
various sizes joined along edges. 
\begin{center}
\includegraphics[scale=.5]{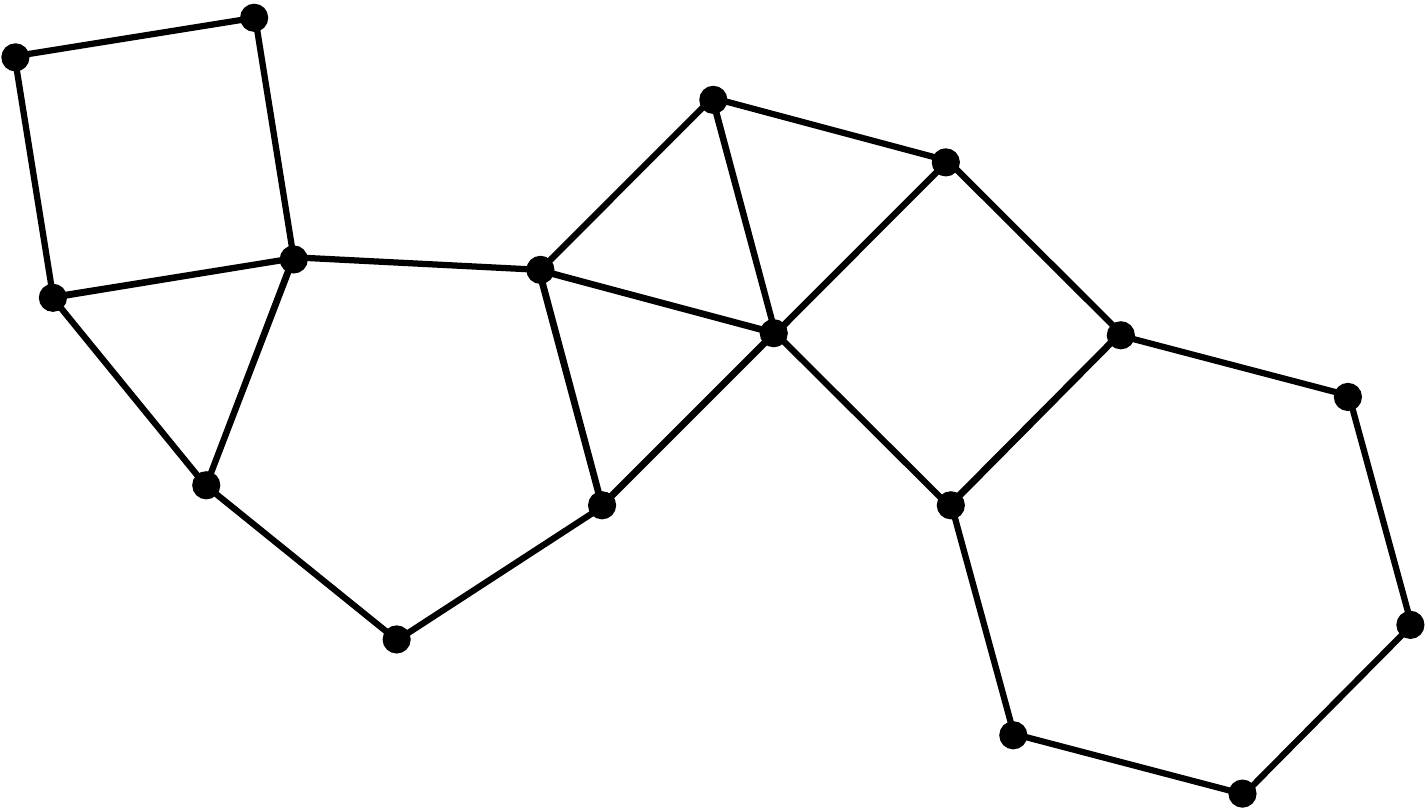}
\end{center}

In the case of the graph hypersurfaces of Feynman graphs, explicit formulae for the
Grothendieck classes of the hypersurface complements for this type of graphs were
obtained in \cite{AluMa3}. However, in the case of the Potts model hypersurfaces,
the recursion relation becomes more complicated: if one combines the formula
for doubling an edge with the formula for splitting it, the resulting class becomes
of the form
$$ (\bT^2 - \bT +1) \{ \cZ_G \} + (\bT-1) \bT \{ \cZ_{G/e} \} + (\bT+1) \{ C^e_G \} + 
(\bT -2) (\bT+1) \{ D^e_G \}, $$
where the term $C^e_G$ is obtained combinatorially from the subgraphs of $G$
that connect the endpoints of $e$, and $D^e_G$ is obtained from paths in 
$G\smallsetminus e$ which do not connect the endpoints of $e$. There is then
no obvious recursive procedure of the type used for either the splitting or the
doubling alone, which takes care of eliminating both of these additional terms.
This remains an interesting case of graphs to investigate.

\section{Estimating topological complexity of virtual phase transitions}

We show here, in the concrete example of the chains of linked polygon graphs
$\bc {(m,k)}G^N$ and in the simpler example of the banana graphs, how one
can use our calculations of classes in the Grothendieck ring to estimate how
the topological complexity of the set of virtual phase transition grows as the
graphs grow in size within the given family.

Good indicators of topological complexity are homologies and cohomologies
and the associated Euler characteristics. In the case of the real algebraic varieties
$\cZ_{G,q}(\R)$, which can in general be singular and non-compact, therefore, 
we can take as a good indicator the Euler characteristic with compact support
$\chi_c(\cZ_{G,q}(\R))$. As we discussed in \S \ref{chicSec}, this is the unique
invariant of real algebraic varieties that is both a topological invariant and a 
motivic invariant. Using the fact that it factors through the Grothendieck ring
$K_0(\cV_\R)$, we obtain the following results.

\begin{prop}\label{chicZGmkNqprop}
Let $\bc {(m,k)}G^N$ be the chain of linked polygons graphs of Definition \ref{GmkN}. Then
the Euler characteristic with compact support $\chi_c(\cZ_{\bc {(m,k)}G^N,q}(\R))$ of the
set of virtual phase transitions $\cZ_{\bc {(m,k)}G^N,q}(\R)$ of the Potts model is given by 
\begin{equation}\label{chicZGmkNq}
(-1)^{mN+kN-k}\left((-1)^N- 2^{kN-k-N} \left(3^{m+1}+1-2^{m+3}\right)^N\right).
\end{equation}
\end{prop}

\proof As we have seen in Example \ref{chicLT}, for real varieties 
$\chi_c(\bT)=\chi_c(\bG_m(\R))=-2$. Using the fact that $\chi_c$ is a ring
homomorphism $\chi_c: K_0(\cV_\R)\to \Z$ we then obtain from \eqref{ZGmkNq}
of Proposition \ref{classGmkNq},
\[
(-1)^{mN+kN-k}\, 2^{kN-k-N} \left(3^{m+1}+1-2^{m+3}\right)^N.
\]
We then use additivity again to obtain 
$$ \chi_c(\cZ_{\bc {(m,k)}G^N,q}(\R)) = \chi_c(\A^{\# E(\bc {(m,k)}G^N)}) -
\chi_c(\{ \cZ_{\bc {(m,k)}G^N,q} \}) , $$
where $\# E(\bc {(m,k)}G^N) = N (m+1) + k (N-1)$ and $\chi_c(\A^1)=\chi_c(\R)=-1$, so that
we obtain \eqref{chicZGmkNq}.
\endproof

We obtain a similar result for the class of graphs $\bc kG^{(m),N}$ of Definition \ref{chainedbananas}, obtained by chains of linked banana graphs.

\begin{prop}\label{chicZGqbanNkm}
Let $\bc kG^{(m),N}$ be the graphs of Definition \ref{chainedbananas}. Then
the Euler characteristic with compact support $\chi_c(\cZ_{\bc kG^{(m),N}}(\R))$
of the set of virtual phase transitions in the model is given by
\begin{equation}\label{chicbanchain}
(-1)^{mN+kN+N-k}\left(1- 2^{k(N-1)} \left(2^m+1\right)^N\right). 
\end{equation}
\end{prop}

\proof The argument is exactly as in the previous case, using the expression \eqref{ZGqchban}
for the classes in the Grothendieck ring. 
\endproof

\subsection{Decision complexity and topological complexity}

There is, in fact, a more technical sense, from the point of view of complexity theory, according
to which the Euler characteristic $\chi_c(\cZ_{G,q}(\R))$ really gives an estimate of the
complexity of the real algebraic variety $\cZ_{G,q}(\R)$ of virtual phase transitions of the
Potts model over the graph $G$. We describe it briefly following the survey \cite{BuCu}.

An algebraic circuit over $\R$ is an acyclic directed graph where each node has in-degree
either $0$ or $1$, or $2$. The nodes of in-degree $0$ are the input nodes and they are
labelled by real variables; the nodes of in-degree $1$ are either output nodes (out-degree
equal to zero) or sign nodes: these are nodes that, to an input $x$ assign output $1$ if
$x\geq 0$ and zero otherwise; and the nodes of in-degree $2$ are labelled by an operation
$+$, $-$, $\times$, or $/$ and are called arithmetic nodes. The size $\sigma(\cC)$ of an 
algebraic circuit $\cC$
is the number of nodes. To an algebraic circuit $\cC$ with $n$ input nodes and $m$ output 
nodes one can associate a function  $\varphi_{\cC}: \R^n \to \R^m$, the function 
computed by the circuit. A decision circuit  is an algebraic circuit with only one output 
node returning values of $0$ or $1$. To each  decision circuit one can associate a 
real (semi)algebraic set $S_\cC =\{ x\in \R^n \,|\, 
\varphi_{\cC}(x)=1 \}$ and, conversely, it is known that each (semi)algebraic set in 
$\R^n$ is realized by some decision circuit. The {\em decision complexity} of a real
(semi)algebraic set $S$ is defined as
\begin{equation}\label{decisioncomplex}
C(S) = \min \{  \sigma(\cC) \, |\, S_{\cC} =S \} .
\end{equation} 
It is known by \cite{Yao} that there is a lower bound on the decision complexity of a
(semi)algebraic set $S$ of the form
\begin{equation}\label{decisionbound}
C(S) \geq \frac{1}{3} ( \log_3 \chi_c(S) - n -4 ),
\end{equation}
in terms of the Euler characteristic with compact support. A similar, more refined estimate
exists in terms of the sum of the Borel--Moore Betti numbers. 

Instead of considering the real algebraic varieties of virtual phase transitions $\cZ_{G,q}(\R)$,
one can look at the actual physical phase transitions. For a finite graph $G$, this 
means considering, in the antiferromagnetic case, the semialgebraic set given by the
intersection of $\cZ_{G,q}(\R)$ with the semialgebraic set 
$S=\{ t\in \A^{\# E(G)} \,|\, -1\leq t_e \leq 0 \}$.  


\end{document}